\PassOptionsToPackage{dvipsnames}{xcolor}
\documentclass[sigconf,authorversion]{acmart} 
\AtBeginDocument{%
  }

\copyrightyear{2026}
\acmYear{2026}
\setcopyright{cc}
\setcctype{by}
\acmConference[CCS '26]{Proceedings of the 2026 ACM SIGSAC Conference on Computer and Communications Security}{November 15--19, 2026}{The Hague, Netherlands}
\acmBooktitle{Proceedings of the 2026 ACM SIGSAC Conference on Computer and Communications Security (CCS '26), November 15--19, 2026, The Hague, Netherlands}
\acmDOI{10.1145/3830454.3846798}
\acmISBN{979-8-4007-2871-6/2026/11}

\usepackage[normalem]{ulem}

\newcommand{\bfs}{\mathbf{s}}

\newcommand{\mbs}{\mathbf{s}}
\newcommand{\mcs}{\mathcal{S}}
\newcommand{\mco}{\mathcal{O}}
\newcommand{\mbd}{\mathbf{d}}

\newcommand{\E}{\mathbb{E}}

\newcommand{\opt}{\textnormal{OPT}}

\let\prod\undefined
\DeclareSymbolFont{cmlargesymbols}{OMX}{cmex}{m}{n}
\DeclareMathSymbol{\prod}{\mathop}{cmlargesymbols}{"51}
\let\sum\undefined
\DeclareMathSymbol{\sum}{\mathop}{cmlargesymbols}{"50}

\newcommand{\eps}{\varepsilon}
\usepackage{amsthm}
\newtheorem*{lemma*}{Lemma}
\newtheorem*{theorem*}{Theorem}
\newtheorem*{corollary*}{Corollary}

\usepackage{balance}
\newcommand{\R}{\mathbb{R}}
\newcommand{\mb}[1]{\mathbf{#1}}

\usepackage{hyperref}
\usepackage{graphicx}
\usepackage{lipsum}
\usepackage[nameinlink]{cleveref}

\theoremstyle{remark}
\newtheorem{remark}{Remark}

\crefname{definition}{Def.}{Defs.}
\Crefname{definition}{Def.}{Defs.}
\crefname{theorem}{Thm.}{Thms.}
\Crefname{theorem}{Thm.}{Thms.}
\crefname{remark}{Rem.}{Rems.}
\Crefname{remark}{Rem.}{Rems.}
\crefname{example}{Ex.}{Exs.}
\Crefname{example}{Ex.}{Exs.}
\crefname{figure}{Fig.}{Figs.}
\Crefname{figure}{Fig.}{Figs.}
\crefname{lemma}{Lem.}{Lems.}
\Crefname{lemma}{Lem.}{Lems.}
\crefname{corollary}{Cor.}{Cors.}
\crefname{corollary}{Cor.}{Cors.}
\Crefname{equation}{Eq.}{Eqs.}
\crefname{equation}{Eq.}{Eqs.}
\Crefname{claim}{Clm.}{Clms.}
\crefname{claim}{Clm.}{Clms.}
\Crefname{conjecture}{Conj.}{Conjs.}
\crefname{conjecture}{Conj.}{Conjs.}

\begin{document}


\title{Adversarial Procurement in Blockchains} 


\author{Maryam Bahrani}
\authornote{Authors listed alphabetically.}
\affiliation{%
  \institution{Ritual}
  \city{Vancouver}
  \state{British Columbia}
  \country{Canada}
}
\email{maryam.bahrani.14@gmail.com}

\author{Michael Neuder}
\authornote{Author supported in part by the Ethereum Foundation: Grant
FY25-2276.}
\affiliation{%
  \institution{Princeton University}
  \city{Princeton}
  \state{New Jersey}
  \country{USA}
}
\email{michael.neuder@princeton.edu}

\author{S. Matthew Weinberg}
\authornote{Author supported in part by NSF CAREER Award CCF-1942497.}
\authornote{During Professor Weinberg’s development of this paper, he participated as an expert witness on behalf of the State of Texas in ongoing litigation against Google (the “Google Litigation”).}
\affiliation{%
  \institution{Princeton University}
  \city{Princeton}
  \state{New Jersey}
  \country{USA}
}
\email{smweinberg@princeton.edu}

\renewcommand{\shortauthors}{Maryam Bahrani, Michael Neuder, \& S. Matthew Weinberg}

\begin{abstract}
    An emerging blockchain protocol design pattern leverages the asymmetry between the computational effort in performing versus verifying tasks. For example, cryptographic validity proofs (e.g., SNARKS) require the prover to expend significant effort demonstrating the correctness of their claim, while the verifiers benefit from extremely easy validation. The operationalization of this paradigm requires efficiently soliciting the performance of expensive tasks in pseudonymous, adversarial environments. We formalize this as a mechanism design question, with a protocol balancing the economic cost of a liveness fault, where the work is not completed, with the payments required to incentivize specific behavior from candidate suppliers. 
    
    We show that the loss of the optimal protocol scales logarithmically in the cost of a liveness fault, scaled up by the adversarial fraction of the network. Further, we find that the optimal equilibria have an intuitive structure, allowing us to provide concrete advice to practitioners. Specifically, in many regimes the optimum designates a single, random node as the primary worker and a committee as a fallback, which is reminiscent of leader-based consensus mechanisms. We also characterize the asymptotic regimes where having negative payments (i.e., slashing in blockchain parlance) is especially helpful.
\end{abstract}



\maketitle

\begingroup
\renewcommand\thefootnote{}
\footnotetext{\textbf{Note: } This is the full version of the paper, which differs from the official ACM Digital Library version only in that it includes the proofs and appendices.}
\endgroup

\fancyhead{}                                     
\fancyhead[LE]{\thepage}                         
\fancyhead[RE]{\shortauthors}                    
\fancyhead[LO]{\shorttitle}                      
\fancyhead[RO]{\thepage}      

\section{Introduction}

In the early days of Bitcoin, every user transacting on the network ran a full node and could easily mine blocks on their personal computers. Once the network became sufficiently valuable, specialization emerged. Miners who invested in better hardware and had a lower cost of electricity were able to capture larger shares of the block rewards. 
Even in its most basic form, this separation between miners and users was possible by leveraging the cryptographically asymmetric task of Proof-of-Work. The miners perform the difficult pre-image search of SHA-256, while the users of the chain can easily confirm that they were included in a valid fork by checking the value of the hashed nonce.

This asymmetry has extended far beyond Proof-of-Work mining. Many core services of modern blockchains are outsourced, allowing for scale and adoption that would otherwise be impossible. 
For example, when sending locally signed transactions, most users will broadcast the signed transaction through a wallet interface that uses an external RPC provider (e.g., Infura \cite{infura}) rather than running a node and gossiping it to the P2P network directly.\footnote{RPC (remote-procedure call) providers expose simple blockchain APIs such as \texttt{eth\_sendTransaction}.} 
Similarly, users often confirm their transaction is included onchain via block explorers rather than running a consensus node locally. Outsourcing is no silver bullet, of course, as relying on any intermediaries introduces choke points and trust assumptions. Still, the ecosystem seems to have converged on the combination of self-custody of private keys while relying on external parties to include and verify transactions as a reasonable trade-off between usability and self-sovereignty. \\

\noindent\textbf{Proposer-Builder Separation: outsourced block-building.} Due to the rise of Maximal Extractable Value (MEV)~\cite{daian2020flash} in the Ethereum ecosystem,\footnote{MEV describes the economic value created by the proposer's unilateral control over the inclusion, exclusion, and ordering of transactions within a block.} even consensus participants now outsource key tasks.
Specifically, creating a value-maximizing block is a complex task where specialized ``builders'' significantly outperform a representative ``proposer'' participating in Ethereum's consensus protocol. Therefore, proposers outsource the job of \textit{creating a high-value block} to builders, and rely on relays as trusted third parties to \textit{verify the value} (and validity) of that block before proposing it in-protocol.\footnote{See \citet{heimbach2023ethereum} for more details on Ethereum's MEV market structure.} This is the first notable example of unbundling the consensus ``duties'' that were homogeneous among validators. See \citet{unbundling-barnabe} for a discussion of how this pattern can extend beyond block building to other parts of the consensus mechanism.


A key aspect of Proposer-Builder Separation (PBS) is that every consensus participant (henceforth, \textit{validator}) can continue to verify the validity of each block. If only specialized actors can check block validity, then Ethereum loses its trustless verifiability, which is a core tenet of the protocol. Indeed, the famous ``Endgame'' post by Vitalik Buterin in 2021 highlights this as a key design criterion: ``block production is centralized, block validation is trustless and highly decentralized, and censorship is still prevented''~\cite{buterin2021endgame}. 

So, the good news is that block building can be outsourced to highly specialized entities under PBS. The bad news is that, \textit{in order to maintain a meaningfully decentralized validator set, every validator must fully verify the contents in each block, which limits the viable throughput of the Ethereum Virtual Machine (EVM).} 
That is, the EVM can proceed no faster than the slowest validator, which at present is set to accommodate, at minimum, a two-core CPU and a 10 MBit/s internet connection~\cite{ethereum_run_a_node}, in order to ensure wide accessibility to Ethereum validation (q.v., some hobbyists run on nodes on a Raspberry Pi~\cite{web3pi}).

But, Ethereum's endgame is not to operate the EVM at the pace of a Raspberry Pi on a home internet connection; instead, Ethereum aims to scale up significantly via a Zero-Knowledge EVM (zk-EVM).\footnote{Again quoting Vitalik, ``In 2027-30, large further gas limit increases, as ZKEVM becomes the primary way to validate blocks on the network''~\cite{zk-evm-vitalik}.}
In the zk-EVM, each block would contain a list of transactions, a claimed newState, and a ZK proof (henceforth, ``proof'') that the EVM's new state after executing the list of transactions from the original state is exactly newState. 
The potential scaling benefits of such an upgrade are immediate; since proof verification is computationally lightweight, Ethereum can maintain a highly decentralized set of validators \textit{who now only verify validity proofs and no longer even need to execute all transactions in a block}, while allowing block content to be optimized by highly-specialized builders and proofs written by highly-specialized provers, thereby enabling the EVM to progress at the pace of the fastest specialized prover rather than that of the slowest executor.\footnote{To clarify: by this, we mean that the current EVM must operate no faster than the slowest validator can \textit{download} and \textit{execute} transactions (to verify block validity). Notice that the zk-EVM is a scaling technology for both \textit{bandwidth} and \textit{compute}, as it allows placing the block contents into Ethereum blobs for validators to sample instead of fully downloading the transaction contents~\cite{nero2026blocks}.
Under this paradigm, the zk-EVM can operate at a faster pace, so long as \textit{someone} in the network can both execute transactions \textit{and produce a correct proof} (because sampling the transaction data and verifying the produced proof is sufficiently lightweight that the entire validator network can do so). Depending on the prover efficiency of the proof system, this may or may not be an improvement -- on one hand, we operate at the pace of the fastest (rather than slowest) participant, but on the other, that participant must now both execute transactions and also write a ZK proof of correctness.} Given the recent advancements in ZK proving speed and architecture \cite{succinct2025hypercube,zksync2025airbender,brevis2025picoprism}, the feasibility of ``real-time proving'' is now a near-term goal~\cite{gold2025realtimeproving}.\\



\noindent\textbf{Verifiable outsourcing to untrusted sources.} The above discussion motivates an exciting application of verifiable outsourcing in the blockchain space: Ethereum validators will outsource transaction execution to specialized ``provers,'' and verify only the proof that the endState is correctly computed. 
So, let us quickly consider a single, centralized, and highly specialized prover who is asked to prove every single zk-EVM block.
Soundness of ZK proof systems guarantees that a malicious prover cannot trick validators into believing an incorrect endState, \textit{but nothing prevents a malicious prover from simply turning off their machines}, causing a liveness fault. That is, if the ecosystem becomes overly reliant on a single prover, that prover can simply stop providing proofs and cause the chain to halt. 
Doing so could be a highly profitable act,\footnote{By accepting bribes from proposers who want to steal MEV from the preceding block's proposer during periods of market volatility, for example.} and relying on reputational harm as the primary defense is completely antithetical to the premise of blockchain systems.\footnote{For example, significant resources in the PBS ecosystem are put towards maintaining relays, protocols (c.f., BuilderNet~\cite{flashbots2024buildernet}), and Trusted Execution Environments (TEEs) that limit the damage a single builder can do, even though there are presently only a small number of dominant builders who are all known by name~\cite{mevboostpics}.}


A next idea is to build a network of (say) $n$ provers, and ask every prover to prove every block. As long as at least one-of-$n$ provers is honest, the proof is delivered. This solution is extremely secure for large $n$, but also correspondingly expensive: we really only need a single correct proof, but are paying for $n$. 
Some overhead is certainly necessary to avoid failure caused by malicious provers, inducing a trade-off. Still, there are many ways to navigate this -- two simple examples include (a) pick a set of $k < n$ provers and pay them to prove and (b) offer a prize of $k$ to be split among all provers who submit proofs (discussed in \Cref{ex:pay-designated,ex:lottery,ex:pay-deliverers}). Are either of these optimal? If so, how does one optimize $k$ as a function of the prover cost versus the harm caused by a liveness fault? \\

\noindent\textbf{Our Work.} This work formulates an adversarial mechanism design question to capture this problem: a protocol has access to $n$ untrusted provers and knows that at least $h < n$ of them are honest. The goal is to design a procurement protocol that minimizes the worst-case economic harm caused by either a liveness fault (if no proof is produced) or overprocurement (if too many proofs are paid for). Importantly, depending on the payment rule, dishonest provers might harm the protocol either by not proving when asked to prove (thus inducing a liveness fault) or by proving when asked not to (thus, perhaps, inducing a higher payment). While prover markets in the upcoming zk-EVM is a core motivation for our model, our results provide insight into any problem of permissionless delegation with concern for liveness faults -- see the end of~\Cref{subsec:related} for a brief discussion of our model's applicability to other domains. 


%

\subsection{Summary of results}\label{subsec:summary}

\textbf{Our main practical contribution is providing an explicit mechanism that we advise practitioners to adopt (\Cref{rem:use-designated}), and arguing that it is feasible and optimal given the order of magnitude of various empirical values observed in the Ethereum ecosystem (\Cref{rem:c-big,rem:proving-cost-small,rem:tau-size,rem:stake-magnitude}).} The theoretical contributions outlined below justify this recommendation.

\Cref{sec:model} formally introduces the model and focuses on the unique requirements of this domain. A single proof is needed, and a penalty of $C$ is incurred if we fail to procure one. 
There are $n$ pseudonymous provers capable of producing the proof for a cost of $1$, $h < n$ of whom are honest. Our goal is to design a mechanism that (perhaps randomly) requests proofs from a subset of participants and offers payments as a function of submitted proofs. This mechanism must be \emph{incentive compatible}: for each prover $i$, provided that all other provers follow the mechanism, prover $i$ optimizes their payoff by following the mechanism as well. 
We then take a worst-case lens and ask, for each such mechanism, if the $n-h$ non-honest provers deviate from the mechanism in the worst possible way, how bad can our total cost be? Note that this includes both monetary payments and the procurement-failure penalty. In blockchain parlance, we assess the performance of a mechanism against a Byzantine attacker who adversarially chooses a set of $n-h$ provers to control fully (whereas the remaining $h$ will follow the mechanism -- see~\Cref{def:loss}).

Our core proposed mechanism is \emph{designated}: pick a single ``leader'' (or ``designate'') from whom to definitely request a proof, and a set of $k$ ``backups'' from whom to request a proof with probability $s < 1$ (and, optimize $k, s$ as a function of $n, h$). 
We prove: (a) for many $n, h$, \emph{including those that well-capture current prover markets}, the designated mechanism is optimal among all mechanisms (\Cref{rem:des-often-optimal}), (b) for all $n, h$, the designated mechanism is an additive $1$-approximation\footnote{Recall that we've normalized the cost of producing a single proof to $1$ -- this means that the designated mechanism is optimal up to the cost of producing a single additional proof.} to the optimal mechanism (\Cref{thm:payment-rule-reduction}). 

We provide an extended discussion for applying these theoretical results to practice (\Cref{subsec:recommendations}). The use of specially nominated leaders, committees, and staking and slashing (negative payments) will be familiar to the blockchain community, due to their frequent use in the design of permissionless consensus mechanisms. We argue, based on reasonable assumptions about: the potential economic impact of a liveness fault (\Cref{rem:c-big}), SNARK proving costs (\Cref{rem:proving-cost-small}), the proportion of honest provers (\Cref{rem:tau-size}), and the amount of capital already staked by Ethereum validators (\Cref{rem:stake-magnitude}), that the optimal designated mechanism should be adopted (\Cref{rem:use-designated}).

\Cref{sec:results} contains significant further exploration of the problem. For example, we consider a weaker adversary (\Cref{lemma:lb}) against which we can fully characterize the optimal mechanism (\Cref{lemma:shape}). Specifically, the optimal mechanism against the weaker adversary is either designated or \emph{symmetric} (a symmetric mechanism requests that each of $k$ providers provide a proof with probability $s$).\footnote{That is, a symmetric mechanism has the ``backups'' from a designated mechanism but not the ``leader.''} We moreover show that whenever a designated mechanism is optimal against the weaker adversary, that same mechanism is optimal against the full adversary (\Cref{lemma:designated-implementable}). We also show that in many regimes where a symmetric mechanism is optimal against the weaker adversary, that same mechanism is optimal against the full adversary (\Cref{lemma:symmetric-anonymous} and~\Cref{lemma:symmetric-implementable}).\footnote{To be explicit: this does not cover all regimes -- there are parameter ranges for which the optimal mechanism against the weaker adversary is \emph{not} implementable against the full adversary. See \Cref{app:regimes-discussion} for discussion on these regimes. We leave as an open problem to fully characterize the optimal mechanism in these remaining regimes (keeping in mind that a designated mechanism is an additive $1$-approximation).}

Finally,~\Cref{sec:results} also analyzes the asymptotic cost of the optimal mechanism as a function of various parameters. Specifically, the worst-case cost of the optimal mechanism scales with $\mco(\frac{n}{h}\cdot \log C)$.\footnote{Recall that $C$ is the ratio of the penalty for a liveness fault to the cost of producing a single proof, and $h/n$ is the fraction of honest provers.} We moreover demonstrate that if the mechanism permits \textit{negative payments}, the optimal solution can greatly improve (\Cref{lem:staking-asymptotics}).  Negative payments are a common motif in blockchain design, where agents post collateral (stake) that is forfeited (slashed) as a result of unexpected behavior.

\subsection{Related work}\label{subsec:related}
Reverse (or procurement) auctions have a long and rich history in the economics literature; see \cite{dimitri2006handbook} as a reference text on the subject. As in traditional (forward) auction theory literature, however, most previous works rely on a fixed set of known identities, where the bidders can misreport with a single identity but cannot arbitrarily introduce new ones. There are a few notable exceptions, starting with \citet{yokoo2004effect}, which introduced ``false-name proof'' mechanisms in the context of combinatorial auctions. Since that paper, the concept has been extended to other domains as surveyed in \cite{conitzer2010using}. Recent work by \citet{pan2024sybil} demonstrates that under the standard quasi-linear utility model of single-item auctions, the only ``Sybil-proof'' mechanism is the second-price auction with symmetric tie breaking. Another recent work by \citet{garimidi2026beyond} studies the single-item procurement setting with Sybils, while also trying to avoid winner-take-all equilibria. The key distinction between our work and these works is that we consider \emph{both} strategic users (by ensuring that the proposed mechanism is incentive compatible) \emph{and} worst-case behavior (by evaluating mechanisms based on their performance with $n-h$ provers adversarially deviating from equilibrium).


Prior work also considers a mix of strategic and malicious behavior, specifically in routing games. Here, \citet{karakostas2003equilibria} first introduces the model, which is further studied in \cite{babaioff2007congestion,roth2008price}. \citet{moscibroda2006selfish} coins the term ``Price of Malice'' and calculated the value for an internet virus game. Our work differs from these in that we study a completely different game.

Specific to prover markets, \citet{wang2025p} introduces a two-sided ZK proof matching protocol to match users to provers. Their algorithm greedily assigns tasks and imposes system-level constraints to address issues such as misreporting capacity, creating Sybils, and failing to deliver when assigned. \citet{ahmadvand2026push0} study the problem of prover orchestration as a systems question, but without modeling the economic decisions of individual provers. 

From the industry side, a few prover markets have launched, and several others are in development. Succinct \cite{roy2024succinct} posts a fixed prize to turn the procurement into a forward auction, in effect running a Tullock contest (see \cite{garimidi2025tullock} for an extended discussion). Axiom \cite{axiom2023announcement}, EigenCloud \cite{eigenlabs2025eigencloud}, Brevis \cite{brevis2025provernet}, Ritual \cite{bahrani2024resonance}, =nil \cite{komarov2023proofmarket}, and RISC zero \cite{boundless2025whitepaper} have all announced development of various ZK coprocessor and prover marketplaces. These proposals have various degrees of formalism; our approach, model, results, and recommendations are novel and provide actionable insights to these teams.

In addition to the specific projects noted above, we conjecture that the paradigm of outsourcing expensive but verifiable work to a set of permissionless and specialized providers will continue to gain adoption for more general tasks beyond correctness proofs for Ethereum blocks. For example, our model captures a basic trusted-execution environment (TEE) implementation of verifiable inference of machine learning models.\footnote{The costly work being done is acquiring the trusted hardware and running the inference in that environment; the low-cost verification being done is checking that the hardware-generated signature attests to the output running in a verifiable way.} 
This type of off-chain proving for onchain verification was coined as a ``ZK coprocessor'' architecture by Axiom in 2023 \cite{axiom2023announcement}. This concept has resurfaced many times since then. Notably, EigenLayer's July 2025 rebrand to EigenCloud \cite{eigenlabs2025eigencloud} and their announcement of verifiable AI inference as a specific target market \cite{alves2026eigenai} fits this model. Similarly, Ritual aims to build a two-sided marketplace for general heterogeneous computational tasks, and they propose using brokers to facilitate the matching of tasks to suppliers in an efficient way \cite{bahrani2024resonance}. Our qualitative lessons (\Cref{subsec:recommendations}) apply equally to any of these domains, and not just to Ethereum prover markets.

\section{Model}\label{sec:model}

\noindent\textbf{Setup.} There is a mechanism designer (``the protocol'') who wants to procure a proof from a set of $n$ players. We normalize the cost of generating a proof to $1$, which we assume is identical across each player and publicly known. We further assume that the validation of the proof is negligible, akin to SNARK verification being nearly constant time (i.e., linear in the public inputs rather than the circuit size) \cite{groth2016size}. The players have pseudonymous IDs (i.e.,~it is possible to request a proof from Player One but not Player Two), but no reputation or commitment power (i.e.,~no player can credibly promise ``if you pay me, I will produce the proof.''). We use the notation $d_i$ to denote the indicator variable for whether prover $i$ delivers a proof.\\

\noindent\textbf{Design space.} The protocol specifies a strategy profile $\mb{s}$ and a payment rule $p(\cdot)$. The strategy profile simply states, for each prover $i$, what is the probability $s_i$ that the prover delivers a proof? The payment rule $p(\cdot)$ specifies, for each prover $i$ and each vector $\mb{d}$, what is the payment $p_i(\mb{d})$ given to player $i$ when the set $\mb{d}$ of proofs is received? For our main results, we must have $p_i(\mb{d}) \geq 0$ for all $i,\mb{d}$. We also consider extensions where each prover $i$ stakes a deposit $B/n$, allowing any $p_i(\mb{d}) \geq -B/n$ to be feasible.\\

\noindent\textbf{Payoffs.} Prover $i$'s ultimate payoff is $p_i(\mb{d})-d_i$. Denote $[n] := \{1, 2, \ldots, n\}$ . The protocol's cost is $\sum_{i\in[n]} p_i(\mb{d})$, the total sum of payments, whenever $\mb{d} \neq \mb{0}$ (i.e.,~at least one proof is produced), and $\sum_{i\in [n]} p_i(\mb{0}) + C$ if no proofs are produced. More formally, we define protocol cost.
\begin{definition}[Protocol cost]\label{def:protocol-cost}
    The \textit{protocol cost} under a payment rule $p$ and delivery profile $\mb{d}$ is the total payment to players, plus the failure penalty $C$ if no proofs are delivered, 
    \begin{align*}
        \text{cost}(p,\mb{d}):=\sum_{i\in[n]} p_i(\mb{d})+ C\cdot  \prod_{i\in [n]} (1-d_i).
    \end{align*}  
\end{definition}

\noindent\textbf{Prover incentives and protocol evaluation.} Provers are risk-neutral and quasilinear, and therefore will only follow the protocol if doing so maximizes their expected utility (in expectation over other provers also following the protocol). Because there is no private information, each prover $i$ can compute their expected payoff $\textsc{DeliverPay}_i$ when submitting a proof (the expected value, when all other players $j$ submit a proof independently with probability $s_j$ to produce $\mb{d}_{-i}$, of $p_i(\mb{d}_{-i};1_i)$, minus the unit cost of delivering), and their expected payoff $\textsc{NoDeliverPay}_i$ when not submitting the expected value, when all other players $j$ submit a proof independently with probability $s_j$ to produce $\mb{d}_{-i}$, of $p_i(\mb{d}_{-i};0_i)$). Therefore, if $s_i > 0$ it must be that $\textsc{DeliverPay}_i \geq \textsc{NoDeliverPay}_i$, and if $s_i < 1$ it must be that $\textsc{NoDeliverPay}_i \geq \textsc{DeliverPay}_i$ (and so if $s_i \in (0,1)$, they must be equal). A protocol is \emph{incentive compatible} if it satisfies these constraints. We refer to $\mcs{(p)}$ as the set of all strategy profiles $\mb{s}$ such that $(\mb{s},p(\cdot))$ is incentive compatible, and say $p$ \textit{implements} $\mb{s}$ if $\mb{s}\in\mcs(p)$.

The protocol designer cares for \emph{worst-case guarantees}. Specifically, they worry that up to $n-h$ of the provers will deviate from equilibrium and behave arbitrarily. 

\noindent\textbf{Brief modeling discussion.} One way to interpret our model is as follows. In steady-state, all provers prove according to the strategy profile $\mb{s}$, which is an equilibrium of $p(\cdot)$. The protocol wants to be robust to a one-time worst-case event where, for whatever reason, up to $n-h$ provers behave arbitrarily rather than as expected. Therefore, incentive compatibility constrains the protocol design (as otherwise, the prescribed equilibrium could not be a steady-state outcome), while the worst-case desideratum guides the analysis. 

\begin{remark}[\textbf{Payment rule as a function of only deliveries}]
Our model restricts payment rules to only depend on the final realized delivery vector $\mbd$. Other payment rules are possible. For example, a protocol could post a bounty that is claimable by the first node to provide a proof (inducing a race), or could pay provers based on the order in which the proofs are submitted. For a decentralized blockchain, order- and time-based payment rules may be less desirable for a number of reasons. First, ordering-based rules induce latency races and may lead to concentration among provers. ZK proving is different than Proof-of-Work in that the fastest prover \textbf{deterministically wins} the generation race because there is no randomness involved in the proving computation. Thus, rewarding the fastest provers would be a winner-take-all outcome and result in a single prover claiming the rewards. Second, in order to pay according to delivery ordering, validators would need to reach consensus over the sequence of deliveries, rather than just the set of deliveries. As timestamps are manipulable and each node will have a differing view of the proof arrival times, this style of payment rule could introduce strategic interactions between the provers and the block producers (as evidenced by the rise of the MEV ecosystem). For these reasons, decentralized blockchains that hope to source proofs from a wide set of providers likely must condition the payments on only the realized set of deliveries, which is much less manipulable.
\end{remark}

\begin{remark}[\textbf{Worst-case analysis as an adversary}]\label{rem:adversary}
    A mathematically equivalent way to interpret the model again has all provers proving according to $\mb{s}$ in steady-state. Then, an adversary corrupts up to $n-h$ provers for one slot and aims to cause as much damage as possible. In the running example of the zk-EVM, the threat model is a one-off attack, where a proposer aims to cause a liveness fault on the slot preceding their own in order to steal the MEV that would've otherwise been paid to the previous slot's proposer.
    Again, incentive compatibility constrains the protocol design, while adversarial robustness guides analysis. Through this lens, our adversary is powerful in the sense that it can choose \textit{which provers} to corrupt on the basis of their pseudonyms (i.e., after learning $\mb{s}$), and cause them to \textit{behave arbitrarily} (i.e., not only by shutting down their prover, but by having them produce a proof they otherwise wouldn't). Our adversary is only limited in that both its corruption and delivery decisions are made without knowing the outcome of the private random coins of honest participants.
\end{remark}

\begin{remark}[Honest nodes and collusion]
As is standard in the consensus literature, we model some the players as \textit{honest} ($h/n$ of them), meaning they will follow the protocol, while the adversary can behave in a worst-case manner. Note that this adversarial model \emph{does} capture collusion between multiple misbehaving nodes (as we treat them all as belonging to the same adversarial set), but it \textit{does not} say anything about collusive behavior between the remaining honest nodes. This can be understood as fixing the fraction of honest and dishonest nodes beforehand and then running the procurement mechanism. In contrast to the consensus literature, our procurement mechanism will function even if less than a super-majority of the nodes are honest, though the expected protocol cost will increase as the honest fraction decreases – see \Cref{fig:protocol-cost} and \Cref{lem:designated-loss-asymptotics}. 
\end{remark}

Based on this interpretation, it is often useful to personify the worst-case analysis as an adversary that is trying to maximally harm the protocol. Denote by $A \subset [n]$ the set of \textit{adversarial} provers, such that $|A| = n-h$, and $H = [n] \setminus A$ as the set of honest (playing the rational equilibrium) provers. Then, the protocol designer evaluates a protocol $\mb{s}, p(\cdot)$ according to following loss function:

\begin{definition}[Loss]\label{def:loss}
    The \textit{loss} of a payment rule $p(\cdot)$ and strategy profile $\bfs$ is 
    \begin{align*}
        \ell(p,\bfs):=
            \max_{A,\mb{d}_A}\Bigl\{\underset{\mb{d}_H \sim\bfs_{H}}{\mathbb{E}} \Bigl[ \text{cost}\bigl(p,\mb{d}_A | \mb{d}_{H}\bigr)\Bigr]\Bigr\},
    \end{align*}
    where $\mb{d}_A | \mb{d}_{H}$ concatenates the adversarial and honest delivery decisions to construct the full delivery vector $\mb{d}$.
\end{definition}

The goal of the designer is therefore to solve the following optimization problem, which minimizes loss over all incentive compatible $(\mb{s}, p(\cdot))$. 

\begin{definition}[Protocol objective]
    The protocol solves
    \begin{align}\label{prog:full}
        \opt1 := \min_{p,\mb{s}\in\mcs{(p)}} \{\ell(p,\mb{s})\}. \tag{PROG1}
    \end{align}
\end{definition}

\begin{remark}[\textbf{Private randomness of honest parties}]
The loss function defined in \Cref{def:loss} takes the $\max$ of the \textit{expectation} over the randomness of the honest player $\mbs_h$. This means that the adversary \textit{does not} observe the delivery decision resulting from the private randomness of the honest parties in the case of a mixed strategy (though the adversary \textit{does observe} the full strategy vector $\mbs$, so if the honest parties are playing a deterministic strategy, $s_i \in \{0,1\}$, then their behavior will be known). This is important for the protocol because it forces the adversary to act under uncertainty. In the running example of procuring zk-EVM proofs, this privacy is implementable with threshold encryption as follows:
\begin{itemize}
    \item During round $n$, require all participants to publish a threshold-encrypted payload containing either the proof (if they chose to do the work) or a dummy string (if they chose not to). 
    \item During the consensus voting for round $n$, each voting validator publishes both their vote indicating if they saw the block in a timely manner \textit{and} their shares of the encrypted payloads from the provers.
    \item By the time a super-majority of votes are made for the block, the decryption results of the proofs are made available. If \textit{all} of the payloads are dummy strings, then the validity of the block cannot be ascertained and the block is skipped (a liveness fault). As long as \textit{a single} payload contains a valid proof, then the decrypted proof is verified to confirm the block and the block is finalized in the chain using a standard consensus mechanism (assuming the super-majority voted that the block was appropriately made available on-time). 
\end{itemize}
Note that we are utilizing the attesting votes to identify both if the payload was published on time and to produce their share of the encrypted proof payloads. Threshold encryption is beneficial in that it shields the proving decisions of the honest parties until a sufficient number of consensus votes have been made. Alternatively, it is also possible to use time-lock puzzles to ensure that there is a gap between when the proofs are committed and when their contents are available. We leave the precise details of prover privacy schemes as implementation design choices, but mention these high-level constructions to highlight the practicality of implementing private randomness in this context.   
\end{remark}

\section{Results}\label{sec:results}
Before analyzing the solutions to this program, the following examples serve to illustrate the difficulty of the design problem concretely. Consider the following basic payment rule.

\begin{example}[Pay a designated subset to deliver]\label{ex:pay-designated}
    Randomly choose a subset of $k$ provers, and pay them $1+\eps$ if they deliver. All $k$ players delivering is a pure-strategy Nash equilibrium.
\end{example}

This naïve payment rule seems promising, but let's consider the worst-case, represented by the Byzantine adversary (\Cref{rem:adversary}) corrupting $a:= n-h$ provers. Since the adversary observes $p$, they will check if $k \leq a$. If so, they can corrupt all $k$ of the designated provers and cause them not to deliver, deterministically causing the liveness penalty of $C$ (which is worse for the protocol assuming $C > k(1+\eps)$). If $k > a$, at least one honest player will deterministically deliver a proof, so the worst-case is the protocol paying the full $k(1+\eps)$ each round. While choosing a high value of $k$ gives a robust mechanism, it is cost-prohibitive to pay for $k$ proofs each slot, especially if the value of $a$ is large. A natural alternative is to commit to a fixed payment and ask players to use a symmetric mixed strategy to determine who delivers.

\begin{example}[Lottery payment rule]\label{ex:lottery}
    Award a fixed prize $P$ randomly to one prover who delivers. Each player delivering with a symmetric probability $s$ is a mixed-strategy Nash equilibrium. 
\end{example}

This mechanism is powerful because it allows the protocol to commit to a fixed prize size \textit{and} leverage the independent randomness of the players to exponentially bound the probability of no one delivering (e.g., $(1-s)^n$).\footnote{This intuition is correct, and we show that this lottery payment rule is a $2-$approximation of optimal (\Cref{lemma:lottery-2-approx}).}
From the perspective of the worst-case attacker, however, this payment rule leads them not to deliver the full set of corrupted provers. This greatly increases the probability of no proof being delivered from $(1-s)^n \rightarrow (1-s)^h$ (only the honest players are flipping coins), and the protocol's expected cost increases correspondingly. Further, the lottery indeed loses some efficiency because it has to \textit{overpay} to incentivize everyone to mix with the same probability $s$ (see \eqref{eq:lottery-prize-size} for the exact equation); one last example payment rule tries to address this.

\begin{example}[Pay deliverers deterministically]\label{ex:pay-deliverers}
    Ask provers to deliver with some uniform probability $s$, paying any prover that delivers $1$. Each prover delivering with a symmetric probability $s$ is a mixed-strategy Nash equilibrium (provers are indifferent between delivering and not because their payment exactly offsets their cost).
\end{example}

This payment rule tries to be more cost-efficient by paying for \textit{precisely} the number of proofs that are delivered. If everyone were honest, this would avoid the overpayment of the lottery payment rule. Consider the adversary; they will compare the protocol's expected cost if \textit{none} of the corrupted provers deliver (which lowers the payment by $a \cdot s$, but also increases the probability of a liveness failure) with the cost if \textit{all} of the corrupted provers deliver (which guarantees no liveness failure, but increases the payment by $a \cdot (1-s)$ -- importantly, this increased payment could be quite significant when $a$ is large and $s$ is small!), choosing the higher of the two. 
These examples illustrate the difficulty of designing for the worst-case, where some provers may over- or under-deliver, depending on the specifics of the payment rule. The remainder of this work formalizes exactly this trade-off.

Solving \eqref{prog:full} directly is complex; see \Cref{app:lp-prog-full} for its expanded form. To start, the number of attacker constraints is exponential in $a$, which arises from the adversary choosing either action for any of the corrupted provers. Further, the protocol has to choose $\mbs$ in such a way that the payment rule achieves a low cost. To make the analysis tractable, we construct a lower bound on $\opt1$ that is only a function of $\mbs$; proof in \Cref{app:lb-proof}.

\begin{lemma}[Lower bound function $g(\mbs)$]\label{lemma:lb}
    Consider the protocol parameterized with $(h,a,n,C)$. Without loss of generality, re-index the players such that $\mbs$ is \textit{decreasing} in $i$. Then define $g(\mbs)$ as,
    \begin{align}\label{eq:two-branch-g}
        g(\mbs) := \max\left\{\sum_{i\in [n]} s_i + C \prod_{i\in[n]} (1-s_i), C \prod_{i=a+1}^n (1-s_i)\right\}.
    \end{align}
    Then the solution to the following program,
    \begin{align}
        \opt{2} := &\min_{\mbs \in [0,1]^n} g(\mbs) \tag{PROG2}\label{prog:g-full}\\
        \text{s.t., } &1\geq s_1 \geq s_2 \geq \ldots \geq s_n \ge 0 \nonumber
    \end{align}
    is a lower bound on the solution to \eqref{prog:full}: $\opt2 \leq \opt1.$
\end{lemma}

\begin{remark}[Interpreting $g$ as a weakened attacker]
    The construction of $g$ as used in the proof of \Cref{lemma:lb} is to \textit{limit the action space} of the Byzantine players. In particular, we force the attacker to choose between playing a \textit{fully honest} strategy (i.e., playing $s_i, \forall i \in A$) or a \textit{fully non-delivering} strategy (i.e., playing $d_i =0, \forall i \in A$). The resulting expected protocol loss is represented as the left and right branches of \eqref{eq:two-branch-g}.
\end{remark}

\begin{remark}[Interpreting $g$ as a powerful mechanism]\label{rem:powerful-mech}
    One alternative view of $g$ is that it is the loss achieved by a mechanism that can immediately detect when the largest $a$ players in the protocol are playing $d_i=0$ and refusing to pay in that circumstance. This detection is represented by the right branch of the $\max$, which is \textit{only} the expected value of the liveness penalty being caused by $h$ smallest indices (and no direct payments). While this seems like a powerful mechanism, we show that in many cases, the bound is achievable (\Cref{rem:des-often-optimal}).
\end{remark}

In order to make the solution of \eqref{prog:g-full} more amenable to analysis, we first prove a structural property about the minimizer of $g$; we show it must occur exactly when the left and right branches of the max are equal. 

\begin{lemma}[Minimizer of $g$ occurs at equality]\label{lemma:minimizer-g-equality}
    The vector $\mbs^*$ that minimizes \eqref{eq:two-branch-g} occurs where the two branches of the max are equal. More formally, we have
    \begin{align*}
        \sum_{i\in [n]} s^*_i + C \prod_{i\in[n]} (1-s^*_i) = C \prod_{i=a+1}^n (1-s^*_i).
    \end{align*}
\end{lemma}

Proof in \Cref{app:minimizer-g-equality-proof}. Using \Cref{lemma:minimizer-g-equality}, we can define a simpler constrained optimization problem for minimizing $g$,
\begin{align}
    \opt{3} :=& \min_{\mbs \in [0,1]^n}  C \prod_{i=a+1}^n (1-s_i)\tag{PROG3}\label{prog:g-constrained}\\
    &\text{s.t., } 1\geq s_1 \geq s_2 \geq \ldots \geq s_n \ge 0 \nonumber \\ 
    &\sum_{i\in [n]} s_i + C \prod_{i\in[n]} (1-s_i) = C \prod_{i=a+1}^n (1-s_i).\nonumber
\end{align}
Since this is exactly equivalent to \eqref{prog:g-full}, we have $\opt3 = \opt2 \leq \opt1$, so the solution to \eqref{prog:g-constrained} is a lower bound on the protocol loss. This transformation simplifies the analysis of the lower bound. Next, we define two ``types'' of equilibria that characterize possible shapes of the solutions to \eqref{prog:g-constrained}.

\begin{definition}[Designated uniform committee equilibrium]\label{def:designated-uniform}
    A strategy vector $\mbs$ is a \textit{designated uniform committee equilibrium} if it has (i) a single deterministic deliverer, (ii) a committee of size $k \leq n-1$ who each mix with the same probability $s\in(0,1)$, (iii) a (potentially empty) set of players who deliver with probability $0$. That is,
    \begin{align*}
        s_i = \begin{cases}
            1 & \text{if $i=1$} \\ 
            s & \text{if $i\in\{2,3,\ldots, k+1\}$} \\ 
            0 & \text{otherwise}.
        \end{cases}
    \end{align*}
\end{definition}
We sometimes use the vector notation for these equilibria
\begin{align*}
    \mb{s} = [1, \underbrace{s, \ldots, s,}_{\text{committee of size $k$}} 0, \ldots , 0].
\end{align*}
Further, we often refer to them simply as ``designated'' equilibria (where the ``designate'' is the player who deterministically delivers by playing $s_i=1$), but will include further context whenever the uniformity or the size of the committee is important.

\begin{definition}[Symmetric committee equilibrium]\label{def:symmetric}
    A strategy vector $\mbs$ is a \textit{symmetric committee equilibrium} if it has (i) a committee of size $k \leq n$ who each mix with the same probability $s\in(0,1)$, (ii) a (potentially empty) set of players who deliver with probability $0$. That is,
    \begin{align*}
        s_i = \begin{cases} 
            s & \text{if $i\in\{1,2,\ldots, k\}$} \\ 
            0 & \text{otherwise}.
        \end{cases}
    \end{align*}
\end{definition}
We sometimes use the vector notation for these equilibria
\begin{align*}
    \mb{s} = [\underbrace{s, \ldots, s,}_{\text{committee of size $k$}} 0, \ldots , 0].
\end{align*}
Further, we often refer to them simply as ``symmetric'' equilibria, but will include further context whenever the size of the committee is important. The minimizers of $g$ are exclusively one of those two shapes.

\begin{lemma}[Characterizing the minimizer of $g$]\label{lemma:shape}
    Let $\mbs^*$ denote the minimizing vector $\mbs^* := \arg\min_{\mbs\in[0,1]}g(\mbs).$ Then $\mbs^*$ is either a designated uniform committee equilibrium or a symmetric committee equilibrium (\Cref{def:designated-uniform,def:symmetric}).
\end{lemma}

Proof in \Cref{app:shape-proof}. \Cref{lemma:shape} gives us a geometric picture of the minimizers of $g$, but remember that this function is only a \textit{lower bound} (\Cref{lemma:lb}) on the protocol loss parameterized just by the strategy vector $\mbs$, without saying anything about the corresponding payment rule $p$.

\begin{definition}[Implementable optimal equilibria]\label{def:implementable}
    An equilibrium vector $\mbs^*$ that is a lower bound on $g$ is \textit{implementable} if there exists a payment rule $p$ such that $\mbs^* \in \mcs(p)$ and $\ell(p,\mbs^*) = g(\mbs).$ 
\end{definition}

An implementable optimal equilibrium achieves the lowest possible protocol cost by making the lower bound tight and thus being a solution to \eqref{prog:full}. The rest of this section analyzes the implementability of the lower bound. 

\Cref{subsec:designated} shows that all designated equilibria are implementable and further, that there exists a designated equilibrium vector that achieves cost no more than $1$ unit higher than the optimal loss (i.e., an additive $1-$approximation). \Cref{subsec:symmteric} demonstrates the conditions under which symmetric minimizers are implementable. \Cref{subsec:transition} examines how the minimizer of $g$ transitions between the two equilibrium shapes and their relative committee sizes. \Cref{subsec:negative-payments} examines the asymptotic scaling of the worst-case loss and how negative payments (i.e., staking and slashing) can reduce it.

\subsection{Designated equilibria}\label{subsec:designated}
If the minimizer of $g$ is designated, we can implement it with the following payment rule recipe. Proof in \Cref{app:designated-implementable-proof}. 

\begin{lemma}[All designated uniform committee equilibria are implementable]\label{lemma:designated-implementable}
    Given a designated uniform committee equilibrium parameterized by committee size $k\leq n-1$ and a mixing probability $s$, there exists a payment rule that implements the equilibria and achieves the minimal protocol loss of $1 + ks$.
\end{lemma}

\begin{remark}[Simplicity of designated mechanism]\label{rem:simplicity-of-designated}
    The payment rule recipe used in the proof of \Cref{lemma:designated-implementable} is simple and extends naturally from our lower bound construction (\Cref{lemma:lb}). In particular, the protocol commits to:
    \begin{enumerate}
        \item[(i)] paying a \textit{fixed amount} $1 + ks$, if at least one proof is delivered, no matter how many copies are produced, and
        \item[(ii)] telling a single prover to deliver with probability $1$, and conditioning the other provers' payment on the deterministic player.
    \end{enumerate}
    This allows the protocol to choose $s$ such that both branches of the max are equal (\Cref{lemma:minimizer-g-equality}) and thus achieve exactly the loss of $g$. 

    From the adversarial perspective (i.e., worst-case analysis \Cref{rem:adversary}), the options are very limited. If they corrupt player $1$ and choose not to deliver $d_1 = 0$, then the protocol will not pay (by (ii) above), and the best they can do is try to cause a liveness penalty by playing a fully non-delivering strategy. If they don't corrupt player $1$, the protocol will certainly acquire a proof, but will only pay a fixed amount (by (i) above). Thus, the payment rule forces the attacker into one of two branches of the lower bound (\Cref{lemma:lb}) and is optimal.
\end{remark}

\Cref{lemma:designated-implementable} tells us that \textit{if} the minimizer of $g$ is designated, then we can implement the optimal mechanism. While this is a positive result, we know that the minimizer of $g$ is sometimes symmetric \Cref{lemma:shape}. \Cref{subsec:symmteric} studies that shape. There is some good news before then, which is that for \textit{any} payment rule $p$ and equilibrium $\mbs$ with loss $\ell(p,\mbs)$, there exists a corresponding payment rule $p'$ and \textit{designated} equilibrium $\mbs'$ such that $\ell(p',\mbs')\leq \ell(p,\mbs)+1$. That is, there is a designated equilibrium that is an additive-1 approximation to the loss of an arbitrary payment rule and equilibrium pair; proof in \Cref{app:payment-rule-reduction-proof}.

\begin{theorem}[Payment rule reduction]\label{thm:payment-rule-reduction}
    Given any payment rule, $p$, and corresponding equilibrium, $\mb{s}\in\mcs(p)$, there exists a modified payment rule, $p'$, and a modified equilibrium, $\mb{s}'\in\mcs(p')$, such that the modified equilibrium is designated and
    \begin{align*}
        \ell(p',\mb{s}') \leq \ell(p,\mb{s}) + 1.
    \end{align*}
\end{theorem}

\begin{corollary}[Designated equilibria are approximately optimal]\label{corr:designated-error}
    By \Cref{thm:payment-rule-reduction}, given the optimal payment rule $p^*$ and optimal equilibrium $\mb{s}^*$, there exists a modified payment rule $p'$ and a correspondingly modified designated equilibrium $\mb{s}'$ such that
    \begin{align*}
        \ell(p',\mb{s}') \leq \ell(p^*, \mb{s}^*) + 1.
    \end{align*}
\end{corollary}

\begin{remark}[On the strength of the approximately optimal construction]
    Our main practical takeaway is advising practitioners to use the optimal designated equilibrium (\Cref{rem:use-designated}), which extends directly from the simplicity (\Cref{rem:simplicity-of-designated}) and approximate optimality (\Cref{corr:designated-error}) of this construction. We also show that the mechanism is asymptotically optimal in $C$ because the additive $1$ error doesn't impact the $\mco(\log C)$ scaling (\Cref{lem:designated-loss-asymptotics}) and that stake can further improve the scaling of the loss (\Cref{lem:staking-asymptotics}).
\end{remark}

Before discussing the asymptotic scaling, however, we cover the non-asymptotic and discrete cases for the optimal loss of the protocol, turning our attention to the symmetric shape for minimizing equilibria of $g$ (\Cref{subsec:symmteric}) and the transition dynamics of the minimizing equilibrium of $g$ (\Cref{subsec:transition}). 

\subsection{Symmetric equilibria}\label{subsec:symmteric}
The other possible form of the minimizer of $g$ is a symmetric equilibrium (\Cref{def:symmetric}). This subsection analyzes the feasibility of implementing such equilibria while achieving the lower bound $g$ (i.e., finding when the lower bound is tight on the symmetric shape). The following lemma allows us to only consider simple anonymous payment rules that only pay based on the total number of proofs delivered. 

\begin{lemma}[Symmetric payment rule reduction]\label{lemma:symmetric-anonymous}
    Any payment rule $p$ that implements a symmetric equilibrium $\mbs$ with a committee size of $k$ can be transformed into an anonymous payment rule $p'$ with the following form:
    \begin{align*}
        p'_i(\mb{d}) &= 
            \begin{cases}
                f_t/t & \text{if } d_i = 1, ||\mb{d}||_1 = t, i \leq k \\ 
                0 & \text{otherwise}.
            \end{cases} 
    \end{align*}
    where $f_t$ is a fixed total prize that the protocol pays under the event that there are exactly $t$ proofs delivered. The symmetric vector is still an equilibrium under the modified payment rule $\mbs \in \mcs(p')$ and $p'$ has weakly lower cost $\ell(p',\mbs) \leq \ell(p,\mbs)$.
\end{lemma}

Proof in \Cref{app:symmetric-anonymous-proof}. \Cref{lemma:symmetric-anonymous} allows us to simply focus on payment rules that set a fixed prize as a function of the number of proofs delivered (i.e., payment rules that are anonymous and don't pay non-deliverers). Consider a symmetric equilibrium with $h$ honest players each mixing with probability $s$. Let $X\sim \text{Binomial}(h,s)$ denote a random variable counting the number of honest deliveries. Further, let $f_i$ denote the total protocol payment given $i$ proofs are delivered, and $f_0 := C$ by definition. Then the optimal payment rule can be written as the solution to the following LP, where $X \sim \text{Binomial}(n-1, s)$:

\begin{align}
    &\min_{f_1,\ldots f_n \geq 0} t \tag{LP1} \label{prog:lp-symmetric}\\
    \text{s.t., \;\;} & \E [f_{i+X}] \leq t \quad \forall i \in \{0,1 \ldots, a\} \tag{attacker delivers i} \label{const:sym-lp-attacker}\\ 
    &\E\left[\frac{f_{1+X}}{1+X}\right] = 1
\tag{equilibrium condition}\label{const:sym-lp-eq}.
\end{align}

See \Cref{app:lp-symmetric-expanded} for a more verbose form of the LP. Each of the \eqref{const:sym-lp-attacker} constraints are defined by the protocols' expected payment over the random distribution of honest payments. Similarly, the \eqref{const:sym-lp-eq} constraint ensures that each player's expected payment is exactly $1$ if they deliver, making them indifferent between delivering and not delivering (both with utility 0), which is a requirement for them to play a mixed strategy. 

We constructed $g$ such that the protocol never pays if it detects that it is being attacked (\Cref{rem:powerful-mech}). When the minimizer of $g$ is designated (\Cref{subsec:designated}), this is easy to do as we can condition the payments on the delivery of the deterministic player. In the symmetric case, it is more involved. First, we define the following linear system, which must be solved to find the values of $f_i$ in order for the lower bound to be tight. 

\begin{definition}[Feasible symmetric linear system]\label{def:feasible-linear-system}
    Consider the values of $f_i$ in \eqref{prog:lp-symmetric}. In order to \textit{detect} that the attacker has played the fully non-delivering strategy (and thus achieve the lower bound on the right branch of $g$), the payment rule never needs to pay if there are fewer than $h+1$ deliveries. More formally, it sets $f_1 = f_2 = \ldots = f_h = 0$. With that, the \eqref{const:sym-lp-attacker} constraint with $i=0$ becomes $C(1-s)^h \leq t$, which is exactly the right branch of $g$, thus we know that the constraint must be tight. Further, in order for the solution to the \eqref{prog:lp-symmetric} to achieve this bound, each of the other \eqref{const:sym-lp-attacker} constraints must also be tight. This allows us to construct the following linear (in $f_{i>h}$) system of equations, which has exactly $a$ equations and unknowns:
    \begin{alignat}{2}
        &(1-s)^h C &&= t\tag{LS1}\label{prog:linear-system}\\
        &s^{h}f_{h+1} &&= t \nonumber \\
        &s^{h}f_{h+2} + h s^{h-1}(1-s) f_{h+1} &&= t\nonumber \\
        &\ldots \nonumber\\
        &\sum_{i=0}^h \binom{h}{i}s^i(1-s)^{h-i} f_{a+i} &&= t.\nonumber
    \end{alignat}
    Each equation in the system corresponds to one of the attacker constraints in \eqref{prog:lp-symmetric} (e.g., the $i^{th}$ equation is the expected protocol payment under $i$ attacker deliveries). In order to achieve the lower bound on $g$, each of the constraints needs to be tight.
\end{definition}

The following lemma gives us an upper bound on $C$ for when symmetric minimizers of $g$ are implementable.

\begin{lemma}[Optimal symmetric implementability]\label{lemma:symmetric-implementable}
    For a given $(h,a,n)$, a symmetric minimizer of $g$ is implementable if 
    \begin{align*}
        C \geq \frac{hn(h+1)^{n-1}}{(h+1)^a - 1}.
    \end{align*}
\end{lemma}

Proof in \Cref{app:symmetric-implementable-proof}. \Cref{lemma:symmetric-implementable} gives us an upper bound on the value of $C$ past which symmetric minimizers of $g$ are implementable (\Cref{def:implementable}) with the payment rule presented in \Cref{lemma:symmetric-anonymous}. Notice that this lemma works for any subcommittee size $k$. However, it is not always the case that these symmetric minimizers are implementable.

\begin{example}[Non-implementable minimizer]\label{example:non-implementable}
    Consider $h=2, n=5, C=7.$ Then $g$ is minimized in the fully symmetric equilibrium (\Cref{def:symmetric}) with $s_1=\ldots = s_n \approx 0.4735$ which results in a $g(\mbs^*) \approx 10\cdot (1-0.4735)^2\approx 2.772$. By \Cref{lemma:symmetric-anonymous}, we know that we only need to check the feasibility of the anonymous symmetric payment rule. Solving \eqref{prog:linear-system}, gives $f_4=-15.133$. Since negative payments are not allowed, this equilibrium is not implementable (\Cref{def:implementable}).
\end{example}

This example demonstrates that our lower bound is sometimes loose. More formally, there exist equilibria $\mbs$, such that $\mbs$ solves \eqref{prog:g-constrained}, but $\nexists p, \mbs \in \mcs(p): \ell(p, \mbs) = g(\mbs).$ This is a limitation of our lower bound function $g$. However, by \Cref{thm:payment-rule-reduction}, we know that the optimal designated equilibrium is an additive $1-$approximation of the true optimal loss. Thus, it follows that if the best implementable symmetric equilibrium (the solution to \eqref{prog:lp-symmetric}) has a lower loss than the optimal designated equilibrium, it too is an additive $1-$approximation of the true optimal. \Cref{app:lottery-payment-rules} points out that a lottery, which is a much simpler symmetric payment rule than the solution to the LP, is a $2-$approximation of optimal.


\Cref{subsec:designated,subsec:symmteric} study the designated and symmetric minimizers of $g$, respectively. \Cref{subsec:transition} explores the \textit{transition points} between these regimes for the minimizer of $g$.

\subsection{Transition dynamics of $g$}\label{subsec:transition}
The lower bound on $g$ is the solution to a non-linear constrained optimization problem, which results in non-obvious properties; consider the relatively small numerical example below. 

\begin{example}[$h=2,n=5$]\label{examples:h2n5}
As we increase the value of $C$, the minimizer $\mb{s}^*$ of $g$ takes multiple shapes:
\begin{itemize}
    \item $C=5 \implies \mbs^* = [1, 0.323,0.323,0.323,0.323]$ \\(designated full). 
    \item $C=7 \implies \mbs^* = [0.473,0.473,0.473,0.473,0.473]$ \\(symmetric full).
    \item $C=10 \implies \mbs^* = [0.926,0.926,0.926,0.926,0.]$ \\(symmetric committee).
\end{itemize}
In words, the shape of the minimizer changes from designated to symmetric. Further, as $C$ increases, the committee size of the symmetric optimizer changes from $5$ to $4$.
\end{example}

From \Cref{lemma:shape}, we know that the minimizer of $g$ is either a designated or symmetric equilibrium, but exactly which it is depends on the specific $(h,n,C)$ values. Geometrically, the solution to \eqref{prog:g-constrained} is subject to:
\begin{itemize}
    \item \textit{Ordering constraints:} Notice that $1 \geq s_1 \geq \ldots \geq s_n \geq 0$ forms an $n-$dimensional polytope. 
    \item \textit{Branching constraint:} By setting the two branches of $g$ equal, the constraint forms an $n-$dimensional surface. 
\end{itemize}
Combining the above observations, we get that the minimizer of $g$ is the minimizing value of the objective on the intersection of the surface and the polytope. In three dimensions, we can visualize this explicitly. 

\begin{figure}
	\centering
	\includegraphics[width=0.9\linewidth]{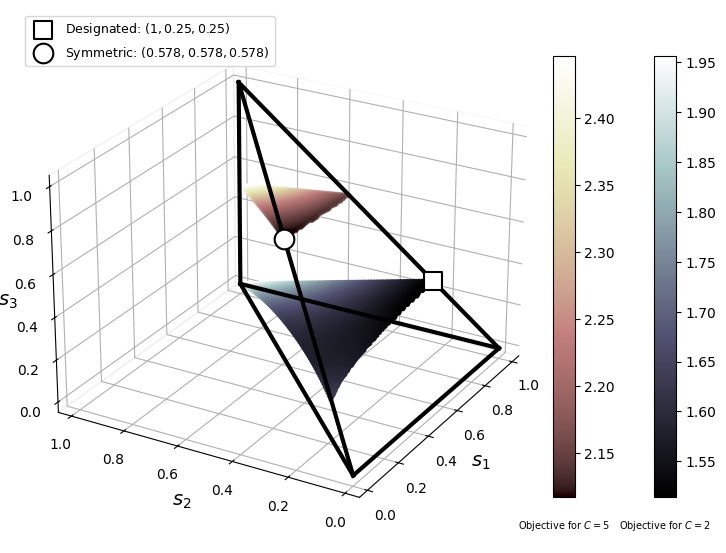}
	\caption{A $3-$dimensional example of the transition dynamics of the solution to \eqref{prog:g-constrained} at $h=1,n=3$ and $C=2$ (lower surface) and $C=5$ (higher surface). The colorbars denote the value of $g$ at $[s_1,s_2,s_3]$, and the circle and square markers indicate the minimizers over the intersection of the surface and the polytope.}
	\label{fig:surface}
\end{figure}

\Cref{fig:surface} shows the $h=1,n=3$ example where the two surfaces ($C=2$, lower \& $C=5$, higher) at $\mbs = [s_1,s_2,s_3]$ with the colorbars indicating the value of $g$. The circle and square markers are the respective minimizers over the intersection of the polytope and each surface. Notice that with $C=2$, the minimizer lies on the edge $[1,s,s]$ (square marker). Conversely, at $C=5$, the minimizer flips to the diagonal $[s,s,s]$ (circle marker). This is exactly what happens in higher dimensions, but the minimizer can also jump to an edge that represents a smaller committee size (\Cref{examples:h2n5}). To analyze the transition between these edges algebraically, we start by defining the minimizers over the respective shapes.

\begin{definition}[Minimizing symmetric and designated losses]\label{def:minimizing-losses-g}
    For a given $(h,a,n,C)$, let $S^*_k,D^*_j$ be the value of $g$ at the minimizing symmetric and designated equilibria, respectively, where $j,k$ denote the optimal committee sizes and $h_j:= j-a,h_k:=k-a$ the resulting number of honest committee members. More formally,
    \begin{align*}
        D_j^* := &\min_{\substack{h_j\in[h]\\ s\in[0,1]}} \big\{1+ (h_j + a-1)s\big\} \\ 
        &\text{s.t. } 1 + (h_j+a-1)s = C(1-s)^{h_j} \\
        S_k^* := &\min_{\substack{h_k\in[h]\\ s\in[0,1]}} \big\{(h_k+a)s + C(1-s)^{h_k+a}\big\} \\ 
        &\text{s.t. } (h_k+a)s + C(1-s)^{h_k+a} = C(1-s)^{h_k}.
    \end{align*} 
\end{definition}

Note that in each case, the optimal committee sizes are $j=(a+h_j), k = (a+h_j)$ respectively, and could be different sizes. The minimizing values consider all possible committee sizes by iterating over the possible values of the size of the honest committee members $h_j,h_k \in [h]$. For each committee size, the optimizer also finds the minimum value of $s \in [0,1]$. In \Cref{examples:h2n5}, we see that the minimizer of $g$ transitions from designated to symmetric. 
\Cref{lemma:desig-small,lemma:sym-big} show that this is the case generally, so long as $a > 1$. In particular, over the interval $C\in(1,\infty)$, the minimizer of $g$ is always designated as $C \to  1^+$ and always symmetric when $C\to \infty$. First, however, we need to handle one special case of $a=1$.

\begin{lemma}[For $a=1$, designated is always optimal]\label{lem:a1}
    For all $C>1, h \geq 1$, if $a=1$ then the optimal designated equilibrium is always better than the optimal symmetric: $D_j^* < S_k^*$.
\end{lemma}

Proof in \Cref{app:a1-proof}. For the remainder of this paper, we will restrict attention to the $a\geq2$ case. First, we show that for small values of $C$, the minimizer is designated.

\begin{lemma}[Minimizer of $g$ is designated for $C < 1+1/a$]\label{lemma:desig-small}
    For very small penalty values $C < 1+1/a$, the minimizing vector $\mbs$ is a designated equilibrium.    
\end{lemma}

Proof in \Cref{app:desig-small-proof}. On the other side of the spectrum, as $C$ gets arbitrarily large, we actually prefer a symmetric equilibrium in the limit.

\begin{lemma}[Minimizer of $g$ is symmetric for $C \to \infty$]\label{lemma:sym-big}
    For large penalty values $C \to \infty$, the minimizing vector $\mbs$ is a symmetric equilibrium. 
\end{lemma}
Proof in \Cref{app:sym-big-proof}. Given that the minimizer changes shape from designated to symmetric on the endpoints of $C\in (1,\infty)$, the smallest value of $C$ where the transition occurs is well defined.  

\begin{definition}[Transition value, $C_t$]\label{def:transition}
    For a given $(h,a,n)$, let $C_t$ denote the smallest $C$ such that the optimal shape of the minimizer of $g$ is symmetric instead of designated. More formally, 
    \begin{align*}
        C_t := \inf_{C\in \R_{>1}} \big\{C: S_j^* < D_k^*\big\}.
    \end{align*}
\end{definition}

\begin{remark}[$C_t$ reduces to root finding on high-degree polynomials]\label{remark:no-closed-form}
    Notice that to find $C_t$, for each candidate $C$ we need to solve for the roots of the $2h$ polynomial equations defined by the constraints in \Cref{def:minimizing-losses-g}. In other words, there is no general algebraic solution to $C_t$ for all $n,h$ pairs.
\end{remark}

\Cref{remark:no-closed-form} tells us that we can't hope for a general expression of $C_t$ for all discrete values of $h,n$. However, the following lemma reduces the analysis to the continuous setting, allowing us to arrive at a more tractable implicit equation.

\begin{figure}
	\centering
	\includegraphics[width=0.95\linewidth]{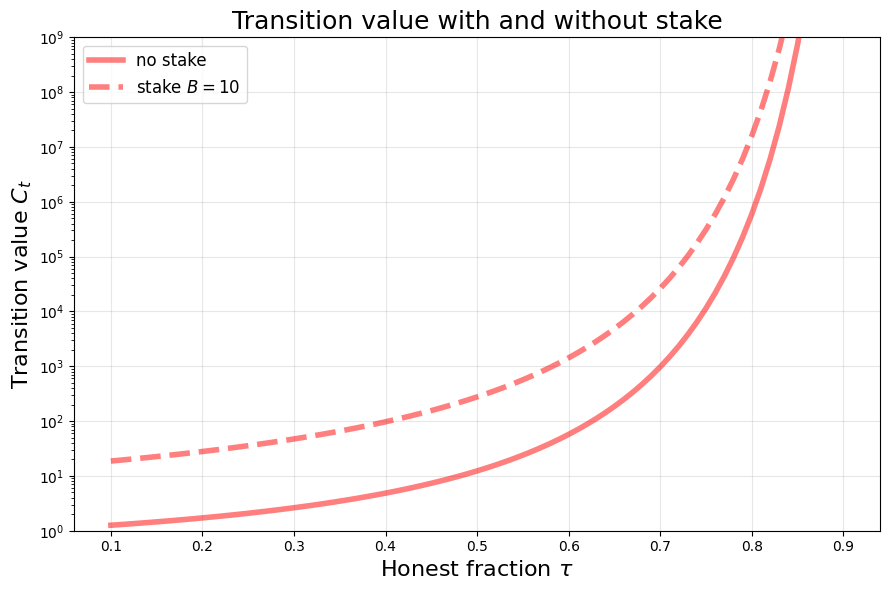}
	\caption{This plot shows $C_t$ (on log scale) as a function of the continuous honest ratio $\tau = h/n \in (0,1)$ when $h,n\to \infty$ as calculated in \Cref{lemma:limit-ct} (solid line). The dashed lines denote the corresponding values with stake $B=10$ (\Cref{subsec:negative-payments}) as calculated in \Cref{cor:limit-ct}.}
	\label{fig:ct-limit}
\end{figure}

\begin{lemma}[Limit behavior of $C_t$]\label{lemma:limit-ct}
As $h,n \to \infty$ with $\tau = h/n \in (0,1)$ as the continuous ``honest proportion,'' the value of $C_t$ at which the optimal symmetric equilibrium achieves a lower loss than the optimal designated equilibrium is:
\begin{align*}
    \lim_{n \to \infty}C_t(\tau) = e^x, \text{ where } 1+x = e^{(1-\tau)x}.
\end{align*}
\end{lemma}

Proof in \Cref{app:limit-ct-proof}. \Cref{fig:ct-limit} shows the analytic bound calculated in \Cref{lemma:limit-ct} as $h,n \to \infty$ as a function of $\tau$ in red. In this continuous setting, we see that $C_t$ is monotone increasing and super-exponential (not the log scale of the y axis) in $\tau$. The discrete setting is not necessarily monotone, but converges to the limit as $h,n$ get large; see \Cref{app:fig-ct-limit-discrete} for a figure with numerical values in the discrete case.

\begin{remark}[$C_t$ growth is super-exponential in $\tau$]\label{rem:super-exp}
    As shown in \Cref{fig:ct-limit}, the growth of $C_t$ is super-linear even with the log-scaled y axis. For a given $\tau$, if we are \textit{below} the red curve, then we know the designated equilibrium is optimal. As such, the super-exponential scaling in $\tau$ is very strong. Intuitively, it tells us that the marginal increase in the honest fraction of provers leads to a much larger corresponding set of $C$ values where designated is exactly the optimal mechanism.
\end{remark}

\begin{remark}[Designated is often optimal]\label{rem:des-often-optimal}
    \Cref{fig:ct-limit} shows that for any $\tau$, if $C < C_t$ (below the red line), the designated equilibrium is optimal. The super-exponentiality in $\tau$ (\Cref{rem:super-exp}) means that for larger regimes of $C$, we have the true optimal mechanism. For example, with no stake and assuming the $\tau=2/3$ honesty threshold from consensus mechanisms, $\forall C < 140$, designated is optimal. Stake, as introduced in \Cref{subsec:negative-payments} below, significantly increases this regime; for the same $\tau=2/3$ with stake $B=10$, designated is optimal $\forall C < 3000$.  
\end{remark}

\begin{figure}
    \centering
    \includegraphics[width=0.9\linewidth]{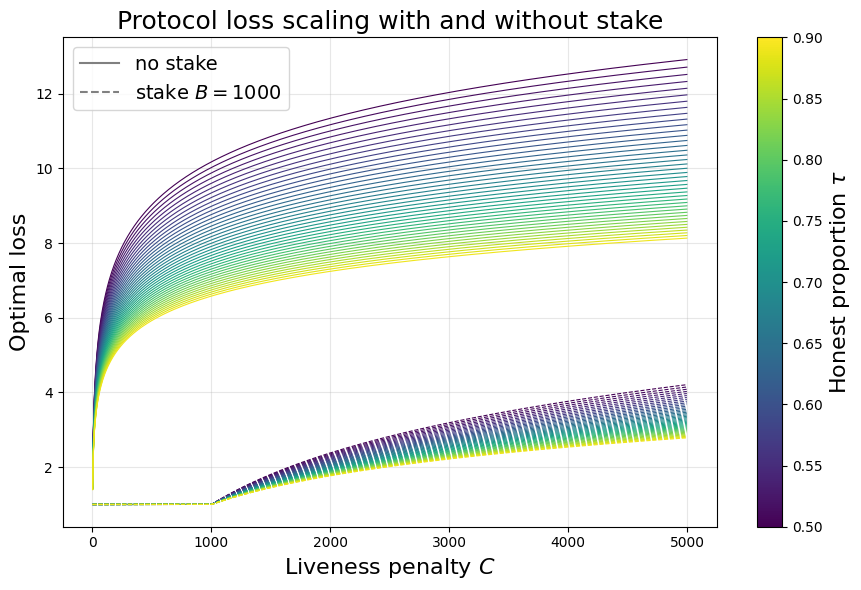}
    \caption{The optimal protocol loss as a function of $C$ for various values of $\tau \in [0.5,0.9]$. The solid lines are without stake. \Cref{lem:designated-loss-asymptotics} shows that this loss scales as $\mathcal{O}\left(\log (C)/\tau\right)$. The dashed lines include stake (\Cref{subsec:negative-payments}) of $B=1000$. For $C \leq 1000$, the cost is constant. For $C> 1000$, we again enter a logarithmic scaling regime; \Cref{lem:staking-asymptotics} details the optimal loss as a function of the asymptotic regimes of $B$.}
    \label{fig:protocol-cost}
\end{figure}

\subsection{Protocol loss scaling and negative payments}\label{subsec:negative-payments}
The value of $C_t$ is just the transition point of when the best symmetric equilibrium is better than the best designated equilibrium (e.g., when our lower bound is tight). We now consider the scaling of the optimal protocol loss as a function of $C$. In words, we want to know how much the protocol should expect to pay in the worst-case, given a liveness penalty of magnitude $C$. We continue in the continuous setting, where the fraction of honest provers is $\tau \in (0,1)$. 

\begin{lemma}[Asymptotic designated loss]\label{lem:designated-loss-asymptotics}
    The loss of the optimal designated equilibrium scales as $\mco\left(\log (C)/\tau\right)$. 
\end{lemma}
Proof in \Cref{app:designated-loss-asymptotics-proof}.

\begin{remark}[Logarithmic scaling of loss]
    \Cref{lem:designated-loss-asymptotics} is a strong result. Even with an adversary that controls any fixed proportion of the provers, we still have asymptotically logarithmic scaling in the magnitude of the liveness penalty $C$. \Cref{fig:protocol-cost} demonstrates this numerically, where the solid lines represent the worst-case protocol loss in the unstaked setting (solid lines) as a function of $C$ and for various values of $\tau$ (colorbar). We introduce the staked (dashed lines) setting below. Even when $C$ gets extremely large, the optimal loss is well controlled and only the coefficient of the log scaling changes with $\tau$. Intuitively, the log scaling comes from the fact that each honest prover is playing an independent mixed strategy. The protocol leverages this to minimize the probability of receiving no proofs.
\end{remark}

To this point, every payment rule we have proposed was constrained to positive payments. We now consider the impact of allowing negative payments.\footnote{Recall that the solution to the \eqref{prog:linear-system} could result in negative values, which we disallowed. This section relaxes that constraint and again performs a lower bound analysis.} Using the nomenclature common in blockchain protocol design, we model this as each player posting a ``stake'' (collateral), which may be ``slashed'' (seized) by the protocol. Of course, the possibility of being slashed is factored into the incentive compatibility of each prover through their equilibrium constraints. That is, any risk of negative payments must be offset by a corresponding positive expected utility from participation.

\begin{remark}[Negative payments as insurance]
    Slashing in our model serves as a form of ``insurance'' for the protocol, which is similar to \texttt{STAKESURE} \cite{deb2024stakesure}. That is, the stake that is slashed by the protocol serves not only as a negative payment for the attacker, but also as a \textit{positive} payment to the protocol to compensate for the damage caused. In our running example of zk-EVM proofs, the proposer of a slot is protected by the collateral posted by the provers so that if their block is not proved (and thus they miss out on the rewards), they are refunded through the slashed collateral. This is a subtle but important distinction from slashing in traditional Proof-of-Stake systems \cite{buterin2020combining}), where the seized collateral is ``burned,'' because there is not an obvious single victim of an attack on the consensus mechanism that could be compensated.
\end{remark}

Note that if the negative payments can be made arbitrarily large, then there is a trivial payment rule that ensures that the protocol achieves the minimal possible cost of $1$. 

\begin{example}[Unbounded stake optimal]
    Require each prover to post a collateral of $C-1$. Designate prover 1 as the ``elected prover'' and pay them according to the following payment rule:
    \begin{align*}
        p_1(\mb{d}) &= 
            \begin{cases}
                -(C-1)& \text{if } d_1 = 0\\
                1 & \text{if } d_1 = 1.
            \end{cases}
    \end{align*}
    Then both outcomes (the prover delivering or not) have the same protocol cost of $1$. 
\end{example}

Obviously, this would be the ideal situation for the protocol designer, but asking each prover to post this large a stake is infeasible in most situations, especially if $C$ is large (discussed in \Cref{rem:c-big}). As such, we analyze the setting where the maximum stake posted by each prover is bounded. We model the aggregate amount of stake across provers that the protocol can elicit as $B$ (for bond), where we are measuring according to the unit, common-value cost of generating a single proof. For example, if $B=100$ and $n=100$, we require each prover to post a collateral of the cost of generating $1$ proof to participate in the system. 
We argue in \Cref{rem:proving-cost-small} that the unit cost of generating a proof should be small, so asking for a potentially large stake is feasible. With $B$ defined, we modify our original lower bound (\Cref{lemma:lb}) to now include the slashed collateral as a refund that \textit{reduces} the protocol cost in the case of an attack.

\begin{corollary}[Staked lower bound function $g_B(\mbs)$]\label{corollary:lb-with-b}
    Extending the lower bound of \Cref{lemma:lb}, consider the protocol parametrized with $(h,a,n,C)$ and with aggregate stake of $B$. Then define $g_B(\mbs)$ as,
    \begin{align*}
        g_B(\mbs) := \max\left\{\sum_{i\in [n]} s_i + C \prod_{i\in[n]} (1-s_i), C \prod_{i=a+1}^n (1-s_i) - B\right\}.
    \end{align*}
    Then for the solution to the following program,
    \begin{align}
        \opt{4} &:= \min_{\mbs \in [0,1]^n} g_B(\mbs) \tag{PROG4}\label{prog:g-full-with-b}\\
        \text{s.t., }& 1\geq s_1 \geq s_2 \geq \ldots \geq s_n \ge 0 \nonumber
    \end{align}
    we have $\opt4 \leq \opt1$.
\end{corollary}
Proof in \Cref{app:lb-with-b-proof}.

\begin{remark}[Correlation of slashing]
    Part of what makes this slashing mechanism powerful is the fact that we can correlate all players' negative payments with the behavior of the attacker. Recall that we are modeling the equilibrium conditions for the non-corrupted players as non-responsive to the threat of an attacker.\footnote{Further work could consider honest players who modify their behavior (through changing their equilibrium conditions) based on the \textit{threat} of the existence of an attacker, but this is out of scope for the present work.} We still think this is the correct model given the fact that liveness attacks (e.g., on Ethereum's zk-EVM) would likely be specific one-off instances. However, with negative payments, this assumption may be harder for provers to abide by. 
    In particular, they may be less willing to post collateral if they know that it could be slashed under an attack, and despite submitting a proof themselves. Still, there is a precedent for correlated penalties; if Ethereum enters an ``inactivity leak'' period (see \texttt{get\_inactivity\_penalty\_deltas} in \cite{ethereum_phase0_beacon_chain}), then every validator receives 0 rewards for the duration of the attack. As such, we think it is a reasonable assumption that if the stakes are small enough, correlated slashing is practical. Further, the correlation can be weakened by designating $k$ provers as deterministic, and only slashing if all $k$ fail to deliver, while only incurring an additional constant overhead of $k$.
\end{remark}

Given $g_B(\mbs)$ (\Cref{corollary:lb-with-b}) and its structural similarity to the lower bound $g(\mbs)$ (\Cref{lemma:lb}), much of the analysis we have established extends directly to the negative payments regime. See \Cref{app:with-stake} for details. Beyond simply calculating the values of $C$ such that we can achieve the optimal cost, we can also make quantitative claims about the sensitivity of the protocol's worst-case loss to the amount of aggregate stake. The following examples show that stake is much more impactful in a high-$C,\tau$ regime. 

\begin{example}[Moderate staking amounts result in large percentage loss reduction if $\tau,C$ are large.]\label{ex:high-tau-c}
    Let $n=100, a= 33$ ($\tau =2/3$), and $C=10,000$. If $B=500$ (5\% of $C$), then the protocol loss drops 47.5\% compared to the no stake environment. See the left half of the table in \Cref{app:stake-sensitivity-table} for more numerical values. 
\end{example}

\begin{example}[Moderate staking amounts yield low percentage cost reduction if $\tau,C$ are moderate.]\label{ex:moderate-tau-c}
    Let $n=200, a=100$ ($\tau =1/2$), and $C=20$. If $B=1$ (5\% of $C$), then the protocol loss drops by 4.5\% compared to the no stake environment. Contrast this with the 47.5\% drop in \Cref{ex:high-tau-c}. See the right half of the table in \Cref{app:stake-sensitivity-table} for more numerical values. 
\end{example}

\begin{remark}
    \Cref{ex:high-tau-c,ex:moderate-tau-c} show that the benefits from acquiring stake that accounts for $5\%$ of the liveness penalty can vary greatly; in this case, the difference in cost reduction was a full order of magnitude higher when $\tau,C$ are larger. \Cref{fig:protocol-cost} shows the impact of stake $B=1000$ on the optimal loss scaling as dashed lines with various $\tau$ (colorbar). We see that across the spectrum of $\tau$ and $C$ values, stake can meaningfully reduce the expected protocol loss.
\end{remark}

Asymptotically, we also consider how the optimal loss scales as a function of $B,C,\tau$.

\begin{lemma}[Asymptotic loss with staking]\label{lem:staking-asymptotics}
    The loss of the optimal designated equilibrium with staking amount $B$, scales as $\mco\left(\frac{\log C - \log(\log C + B)}{\tau}\right)$. Observe that this implies the following regimes
    \begin{alignat*}{2}
        B &= O(1) &&\implies loss = \mco\left(\log (C) / \tau \right) \\ 
        B &= \Theta(C/\textnormal{polylog } C) &&\implies loss = \mco\left(\log (\log C)/\tau\right)\\ 
        B &= \Theta(C) &&\implies loss = \mco\left(1/\tau\right).
    \end{alignat*}
\end{lemma}
Proof in \Cref{app:staking-asymptotics-proof}. \Cref{ex:high-tau-c,ex:moderate-tau-c,lem:staking-asymptotics} show that stake can significantly reduce the protocol loss. We use these results to motivate our concrete takeaways for practitioners in \Cref{subsec:recommendations}.

\section{Discussion}\label{sec:discussion}
\Cref{sec:results} studies the minimizers and transition dynamics of the lower bound $g$ (\Cref{lemma:lb}) on $\opt1$ \eqref{prog:full}. The lower bound is highly tractable, and we show that it is tight in some regimes (\Cref{lemma:designated-implementable,lemma:symmetric-implementable}) and that we can get good approximations in all regimes (\Cref{thm:payment-rule-reduction,lemma:lottery-2-approx}). Further, we showed that no closed form exists for when the transition occurs (\Cref{remark:no-closed-form}), but we analyzed the limit behavior and gave numerical and asymptotic bounds (\Cref{lemma:limit-ct,lem:staking-asymptotics}). 

This section builds on these results by considering \textit{only implementable} symmetric and designated equilibria, conjecturing that the true solution  \eqref{prog:full} is the better of the optimal \textit{and} implementable symmetric or designated equilibria. Intuitively, this conjecture would say that the minimizing shapes of $g$ are the minimizing shapes of \eqref{prog:full} more generally.

In the symmetric case, the best implementable payment rule is exactly the solution to \eqref{prog:lp-symmetric}. \Cref{lemma:symmetric-implementable} showed that this is equal to the lower bound on $g$ if and only if $f_1 = \ldots = f_h = 0$, but the solution to the LP is the optimal implementable symmetric payment rule more generally. 

In the designated case, all equilibria are implementable (\Cref{lemma:designated-implementable}). Thus, to find the optimal implementable equilibrium, we consider a modification of \eqref{prog:g-constrained}, where we simplify the program by fixing $s_1=1$. Also, by the same argument we make for the shape of $g$ in \Cref{lemma:shape}, the solution to the modified program will be a designated uniform committee equilibrium (\Cref{def:designated-uniform}). Thus, we can write the simplified program for the full set of $n$ players as:

\begin{align}
    \opt{5} &:= \min_{\mbs \in [0,1]}  C (1-s)^h\tag{PROG5}\label{prog:constrained-designated}\\
    \text{s.t., } &1 + (n-1)s = C (1-s)^h \nonumber
\end{align}

We now define the minimizing \textit{implementable} symmetric and designated losses (denoted $\hat{S}_k,\hat{D}_j$ resp.), which is very similar to \Cref{def:minimizing-losses-g} ($S^*_k,D^*_j$), but with the enforcement of implementability rather than searching over the general lower bound $g$.

\begin{definition}[Minimizing implementable symmetric and designated losses]\label{def:minimizing-losses-implementable}
    For a given $(h,a,n,C)$, let $\hat{S}_k,\hat{D}_j$ be the values of the minimizing implementable symmetric and designated equilibria, respectively, where $j,k$ denote the respective optimal committee sizes and $h_j:= j-a,h_k:=k-a$ denote the resulting number of honest committee members. More formally,
    \begin{align*}
        \hat{D}_j := &\min_{\substack{h_j\in[h]\\ s\in[0,1]}} \big\{\ref{prog:constrained-designated}(h_j)\big\}, \quad
        \hat{S}_k :=\min_{\substack{h_k\in[h]\\ s\in[0,1]}} \big\{\ref{prog:lp-symmetric}(h_k)\big\},
    \end{align*} 
    where $\ref{prog:constrained-designated}(h_j), \ref{prog:lp-symmetric}(h_j)$ denote the solutions to the programs parameterized by committee sizes of $j=a+h_j, k=a+h_k$.
\end{definition}

\begin{conjecture}[Optimal is designated or symmetric]\label{conj:minima}
    The solution to \eqref{prog:full} is the smaller of the implementable minimizing designated and symmetric equilibria. More formally, 
    \begin{align*}
        \opt1 \overset{?}{=} \min\left\{\hat{D}_j,\hat{S}_k\right\}.
    \end{align*}
\end{conjecture}

This conjecture is motivated through the analysis of the lower bound $g$ in \Cref{sec:results}, where we are able to explicitly pin down the shape to designated and symmetric minimizers of $g$ (\Cref{lemma:shape}). \Cref{app:regimes-discussion} discusses this conjecture and its relationship with the lower bound of $g$ in detail. \Cref{subsec:counter-examples} considers several conjectures that would make the analysis of the general optimal mechanism or lower bound more structured and tractable. 
However, we also show a counter-example to each conjecture, further justifying the complexity of the optimal mechanism and the value of the approximately optimal construction (\Cref{thm:payment-rule-reduction}). 

We leave \Cref{conj:minima} as an open question for future work. Based on our estimates of the relative size of $C,\tau,B$ and proving costs (\Cref{rem:c-big,rem:tau-size,rem:stake-magnitude,rem:proving-cost-small}), our main recommendation is to use designated equilibria (\Cref{rem:use-designated}). The following subsection argues that the practicality and near-optimality of the best designated equilibrium make it a powerful tool.

\subsection{Recommendations to practitioners}\label{subsec:recommendations}
The results in \Cref{sec:results} are readily interpretable and have immediate applicability for blockchain protocol design today. This section distills these takeaways. First, we begin with a set of observations that help frame the results by considering the magnitude of the model's variables based on empirical reference values.

\begin{remark}\label{rem:c-big}
    \textbf{$C$ is large.} We believe that protocol designers should model the cost of a liveness penalty as being quite large. For example, in the zk-EVM model, the ``economic loss'' of a liveness penalty is a missed slot because the fork-choice rule will reject a block without an accompanying proof. Even though blocks on average have relatively low MEV payments, there are occasional, extremely large rewards (c.f., a recent block \cite{mevproposerbot} with 189 \texttt{ETH} of MEV $\approx$ \$440,000 at April 2026 prices). The attack model in this situation is the proposer \textit{of the following slot}, who stands to benefit greatly from the previous block not being included and thus capturing the MEV for themselves. To be robust during these high-volatility periods, protocols should consider that a liveness penalty may have a very high economic cost.
\end{remark}

\begin{remark}\label{rem:proving-cost-small}
    \textbf{Proving cost is relatively small.} Remember that our model uses the common-value setting, where we assume a unit proving cost compared to the liveness penalty $C$ (i.e., if $C=100$, then the liveness penalty is $100\times$ the cost of generating a proof). Given the size of $C$ described above, we think designers should consider the proving cost as a small fraction of it. For example, in the zk-EVM use-case, we see prices for proofs on today's blocks are on the order of \$0.1 \cite{ethproofs}, though this number may go up based on Ethereum's resulting scaling and doesn't factor in the upfront cost of purchasing hardware. For other verifiable workloads, the ratio between the proving cost and the liveness penalty might be smaller. For example, if the protocol is trying to procure a verifiable inference on an LLM, the cost of running it (e.g., on a TEE) may be relatively higher compared to the cost of not getting the result delivered to the user. In this regime, the protocols may choose different parameters. Still, we think for a majority of blockchain protocols, liveness will be a first-order concern, so the large-$C$ paradigm is the focus of this work.
\end{remark}

\begin{remark}[\textbf{Common-value proving cost}]
    We model proving cost as publicly known and identical between providers, which is reasonable for the zk-EVM use case we describe throughout the paper. In the more general setting where suppliers have very different costs, eliciting this information \textit{before} deciding who to ask to perform work might lead to a more efficient mechanism. For example, the protocol could run an incentive-compatible procurement auction to solicit truthful reports on the cost of generating proofs for the honest players. This model would be more complex under the adversarial setting we explore, as the byzantine adversary would manipulate both the bids reported \textit{and} the delivery decisions of the suppliers to be worst-case for the protocol. As discussed in \Cref{subsec:related}, very few mechanism design papers model adversarial participants,\footnote{There is some work on cases where suppliers may \textit{fail} to deliver \cite{porter2012mechanism}, but those models account for random crash faults, where the delivery might not occur, and not for the explicit worst-case bidding behavior of an adversary.} and thus we leave the analysis for this extension as a good direction for future work. 
\end{remark}

\begin{remark}\label{rem:tau-size}
    \textbf{Honest fraction $\tau$ is at least moderate.} The honest proportion of the prover set $\tau$ plays an important role in our analysis. Unlike the consensus literature that typically requires a $2/3$ honesty assumption, our mechanism is general for any value of $\tau$. Of course, as $\tau \rightarrow 0^+$, the worst-case protocol cost will approach $C$ (at $\tau =0$, the adversary can deterministically cause a liveness penalty), but we believe that assuming a non-zero fraction of the provers are honest is reasonable. One justification for this assumption is that staking provides a budget constraint for the attacker. For example, with $\tau = 0.5$, we could view the attacker budget as being able to stake (or bribe) up to 50\% of the prover set. For larger values of $\tau$ (à la the $2/3$ honesty assumption of Byzantine consensus), the results are even stronger. 
    
    In \Cref{fig:ct-limit}, we see that $C_t$ is super-exponential in $\tau$, meaning for moderate to large values of $\tau$, the minimizer of $g$ is designated for many values of $C$ and thus the lower bound is tight (\Cref{lemma:designated-implementable}). In \Cref{fig:protocol-cost}, we see that the $\tau$ coefficient on the log scaling meaningfully reduces the protocol cost in absolute terms. Asymptotically, the optimal unstaked protocol loss scales as $\mco\left(\log (C)/\tau\right)$. If $\tau = \Omega(1)$, the logarithmic scaling of the loss is preserved. If $\tau$ does approach zero, but has scaling $\tau = \Omega(1/\text{polylog } C)$, then the loss scales faster, but only as $\mco(\text{polylog } C)$.
\end{remark}

\begin{remark}\label{rem:stake-magnitude}
    \textbf{The aggregate staking amount $B$ is large.} In our model, we consider the total amount of stake put up by the provers (\Cref{corollary:lb-with-b}). Asymptotically, we see that if $B = \Theta(C)$, the protocol loss scales constantly if $\tau$ is constant in $C$ (\Cref{lem:staking-asymptotics}). While $B$ being the same order as $C$ sounds far-fetched in light of the magnitude of $C$ (\Cref{rem:c-big}), consider that $B$ is the \textit{total} amount of stake rather than an individual's contribution. As such, even if $C$ is on the order of $\$10mm$ worth of capital (1.5 orders of magnitude higher than discussed in \Cref{rem:c-big}), then asking the prover set in aggregate to stake that much is plausible. For example, if there were 100 provers, each would have to stake $\$100,000$. In Ethereum today, the minimum validator balance is 32 \texttt{ETH} $\approx \$74,000$ at April 2026 prices, meaning the provers would be asked to stake a mere 33\% more than the solo stakers already do today.
\end{remark}

\begin{remark}\label{rem:use-designated}
    \textbf{Practitioners should use the optimal designated uniform committee equilibrium (\Cref{def:designated-uniform}).} 
    This confident recommendation is grounded in the empirical estimates from the Ethereum ecosystem of: the magnitude of MEV spikes (\Cref{rem:c-big}), the current average cost of generating SNARKS for EVM blocks (\Cref{rem:proving-cost-small}), reasonable honesty assumptions (\Cref{rem:tau-size}), and the dollar amount of capital that validators currently stake (\Cref{rem:stake-magnitude}). With this context, the theoretical results are extremely compelling. \Cref{lemma:designated-implementable} shows us that \textit{any} designated equilibrium is implementable; \Cref{rem:simplicity-of-designated} notes the simplicity and interpretability of the payment rule. \Cref{thm:payment-rule-reduction} demonstrates that the optimal designated loss is at most the cost of generating a single proof larger than the optimal loss; in light of today's proving costs, this may be less than a dollar (\Cref{rem:proving-cost-small}). \Cref{lem:staking-asymptotics} shows that if the aggregate stake is of the same order as $C$, which \Cref{rem:stake-magnitude} argues is reasonable, then the optimal protocol will only incur a constant cost.

    Beyond the optimality of protocol loss considerations, designated equilibria are readily interpretable and have many similarities to familiar permissionless consensus protocols that underpin today's blockchains. The analogs are striking. In both, a single node is elected each round to have a ``special role'' (e.g., the leader in Tendermint \cite{buchman2018latest}). Further, each suggests a committee that ``checks'' the power of the single node (e.g., the attesting committee in HLMD GHOST \cite{buterin2020combining}). Lastly, each leverages staking as a form of accountability (e.g., ``accountable safety'' in Casper \cite{buterin2017casper} and ``accountable liveness'' \cite{lewis2025accountable}). As such, we believe the blockchain community will be open to seriously considering our design as a viable architecture for implementing prover markets and general procurement tasks.
\end{remark}

Based on the above, we strongly advocate for the usage of the optimal designated equilibrium. Namely, the protocol designed can find $\hat{D}_j$ (\Cref{def:minimizing-losses-implementable}) and use the resulting optimal committee size $j$ and mixing probability $s$. However, if the protocol designer only wants to consider symmetric equilibria, then we also provide pragmatic advice when considering the optimal symmetric approach.  

\begin{remark}
    \textbf{How to implement symmetric equilibria.} If $C$ is exceedingly large, we may be in the symmetric regime (e.g., $C > C_t$ \Cref{def:transition}) and the optimal symmetric equilibrium may be implementable (\Cref{lemma:symmetric-implementable}). In that case, the lower bound is tight, and this symmetric equilibrium is optimal. Alternatively, the protocol designer might care about the \textit{type} of permissionlessness in the architecture. For example, if the protocol \textit{cannot} require stake and provers are free to come and go at will in any round (the \textit{fully permissionless} environment as defined in \cite{lewis2023permissionless}), then a lottery mechanism (\Cref{def:lottery}) is a good fit and still achieves a $2-$approximation of the optimal loss (\Cref{lemma:lottery-2-approx}), leading to the same asymptotic scaling.
\end{remark}

\section{Conclusion}
This paper studies how blockchain protocols can leverage the asymmetry between performing computational work and verifying it. We argue that a robust procurement mechanism should withstand a Byzantine attacker and model the protocol as a mechanism design problem. Given the economic loss caused by a liveness penalty, the protocol aims to minimize its loss under an attacker behaving arbitrarily with a portion of the nodes in the network. Our results are theoretical and have clear insights for practitioners.



\bibliographystyle{ACM-Reference-Format}
\bibliography{refs}

\appendix 

\section{Open Science} 


\begin{itemize}
\item Repository containing the full set of paper plots:\\ \href{https://github.com/welikecrypto4844/adversarial-procurement-in-blockchains}{https://github.com/welikecrypto4844/adversarial-procurement-in-blockchains}
\end{itemize}



\section{Section 3 Appendices}

\subsection{Expanded form of \eqref{prog:full}}\label{app:lp-prog-full}
It is helpful to consider the protocol objective \eqref{prog:full} as a two-layer optimization. Once the equilibrium vector $\mbs$ is set, then the minimal protocol loss (encoded as payments made under different delivery vectors) can be solved as a linear program.

\begin{definition}[Two-level optimization]
    First, consider an arbitrary fixed strategy vector $\mbs \in [0,1]^n$. The protocol can then solve the following linear program (LP) in $p$.
    \begin{align}
        L(&\mbs):= \min_{p, t\geq0} t \tag{LP2}\label{eq:lp1}\nonumber \\
        s.t. \quad&\underset{\mb{d}_H \sim \mb{s}_H} {\mathbb{E}}
        \Big[\sum_{i=1}^np_i(\mb{d}_A\mid\mb{d}_H)+C\prod_{i=1}^n(1-d_i)\Big]
        &\le&t&\forall&(A,\mb{d}_A)\tag{attacker}\label{const:att}\\
        &\underset{\mb{d}_{-i} \sim \mb{s}_{-i}}{\mathbb{E}}\big[p_i(1\mid\mb{d}_{-i})-p_i(0\mid\mb{d}_{-i})\big]
        &\ge&1& \quad \forall&i\text{ if }s_i>0\tag{equil.\ 1}\label{const:eq1}\\
        &\underset{\mb{d}_{-i} \sim \mb{s}_{-i}}{\mathbb{E}}\big[p_i(1\mid\mb{d}_{-i})-p_i(0\mid\mb{d}_{-i})\big]
        &\le&1&\forall&i\text{ if }s_i<1\tag{equil.\ 2}\label{const:eq2}
    \end{align}
    This LP uses an epigraph variable to define the upper bound constraints on the protocol loss given any attacker set and action pair \eqref{const:att}. Further, the equilibrium constraints (\ref{const:eq1},\ref{const:eq2}) ensure that the strategy is an equilibrium implemented by $p$, i.e., $\mbs \in \mcs(p)$. If $s_i \in (0,1)$, then these constraints are both tight at $1$, and otherwise, deterministically delivering or not delivering are dominant strategies. Given this inner LP, we define the outer optimization problem,
    \begin{align}
        \opt1 = \min_{\mbs \in [0,1]^n} L(\mbs),
    \end{align}
    and we have that the solution to this program is the solution to \eqref{prog:full} and minimizes the protocol's objective.
\end{definition}

\subsection{Proof of \Cref{lemma:lb}}\label{app:lb-proof}
\begin{lemma*}[Lower bound function $g(\mbs)$]
    Consider the protocol parameterized with $(h,a,n,C)$. Without loss of generality, re-index the players such that $\mbs$ is \textit{decreasing} in $i$. Then define $g(\mbs)$ as,
    \begin{align}
        g(\mbs) := \max\left\{\sum_{i\in [n]} s_i + C \prod_{i\in[n]} (1-s_i), C \prod_{i=a+1}^n (1-s_i)\right\}.
    \end{align}
    Then the solution to the following program,
    \begin{align}
        \opt{2} := &\min_{\mbs \in [0,1]^n} g(\mbs) \tag{PROG2}\\
        \text{s.t., } &1\geq s_1 \geq s_2 \geq \ldots \geq s_n \ge 0 \nonumber
    \end{align}
    is a lower bound on the solution to \eqref{prog:full}: $\opt2 \leq \opt1.$
\end{lemma*}
\begin{proof}
    First, consider a restricted attacker that chooses between two options: behaving honestly (e.g., playing the strategy $s_i$ for all corrupted indices $i \in A$) and fully not delivering, with the largest $a$ values of $\mbs$ (e.g., setting $d_i = 0$ for $d \in [a]$). 
    
    If the attacker plays honestly, the protocol must pay at least $\sum_{i\in[n]}s_i$ to the provers in order for the mechanism to be individually rational. Further, there may be some non-zero probability of the event that all of the provers don't deliver and the protocol incurs the penalty $C$. In particular, since the provers are mixing independently, the probability that none of them deliver is $\prod_{i\in[n]}(1-s_i)$ (recall that $s_i$ is their mixing probability \textit{for} delivering). Thus, the left branch of the $\max$ in $g(\mbs)$ is exactly the protocol's expected cost under an attacker that plays the fully honest strategy. 

    By similar reasoning, if the attacker doesn't deliver ($d_i=0$) with the largest $a$ values of $s_i$, then the protocol will incur the penalty with probability equal to that of the smallest $h$ indices (the remaining honest players) not proving $\prod_{i=a+1}^n(1-s_i)$. Thus, the right branch of the $\max$ is a lower bound on the protocol's cost under a fully non-delivering attacker. 

    Since the $\max$ of these two branches is always a viable action for the attacker and the protocol loss (\Cref{def:loss}) is defined over \textit{all attacker} actions and sets, this is a lower bound on $\opt1$.
\end{proof}

\subsection{Proof of \Cref{lemma:minimizer-g-equality}}\label{app:minimizer-g-equality-proof}
\begin{lemma*}[Minimizer of $g$ occurs at equality]
    The vector $\mbs^*$ that minimizes $g$ occurs where the two branches of the max are equal. More formally, we have
    \begin{align*}
        \sum_{i\in [n]} s^*_i + C \prod_{i\in[n]} (1-s^*_i) = C \prod_{i=a+1}^n (1-s^*_i).
    \end{align*}
\end{lemma*}
\begin{proof}
    We show that the optimizer can never be strictly in either branch. Recall that,
    \begin{align*}
        g(\mbs) := \max\left\{\sum_{i\in [n]} s_i + C \prod_{i\in[n]} (1-s_i), C \prod_{i=a+1}^n (1-s_i)\right\},
    \end{align*}
    and define $L,R$ as the left and right branches of the max:
    \begin{align*}
        L(\mbs) &:= \sum_{i\in [n]} s_i + C \prod_{i\in[n]} (1-s_i) \\
        R(\mbs) &:= C \prod_{i=a+1}^n (1-s_i).
    \end{align*}
    Assume, towards a contradiction, that $\mb{s}^*$ (the minimizing $\mb{s}$ of $g$) is strictly in the right branch of the max $L(\mbs^*) < R(\mbs^*)$. First note that $s_i^* < 1$ for all $i \in \{a+1, \ldots, n\}$, because if any $s^*_i = 1$, then $R(\mbs^*)$ would be zero and couldn't be the max because $L(\mbs^*) > 0$. 
    
    Let $j$ denote the first index such that $s_j < 1$ and from above, $j \leq a+1$. Create a modified vector $s'$ such that $s_k' = s_k^* + \eps$ for $k \in [j, \ldots, n]$. By continuity and the strictness of the inequality, we know that there exists an $\eps$ such that we are still in the right branch: $L(\mbs') < R(\mbs')$. Further, the ordering constraints are respected because all values of $\mbs$ that change are modified by the same amount, $\eps$. Lastly, since all values of $s_i$ for $i \in \{a+1, \ldots, n\}$ increased, we have $R(\mbs') < R(\mbs^*)$, which is a contradiction on $\mbs^*$ being optimal.   

    Now assume, again towards a contradiction, that $\mbs^*$ (the minimizing $\mbs$ of $g$) is strictly in the left branch of the max: $L(\mbs^*) > R(\mbs^*)$. Notice that $L(\mbs)$ is order invariant, because the sum and product terms are taken over all $i \in [n]$. As such, we consider the minimization of $L(\mbs)$ over the unit hypercube $\mbs \in [0,1]^n$. We claim that with $C>1$, the solution to this minimization is exactly the standard basis vectors, which under the ordering constraints implies $\mbs^* = [1, 0,\ldots, 0]$. We show this through casework on the other possible solutions: \begin{itemize}
        \item If $s_i=1$ for some $i$ and $s_{j} >0$ for some $j \neq i$, decrease the value of $s_j$. The sum term of $L(\mbs)$ decreases, while the product is still zero. 
        \item If no coordinate has a value of $1$ and there exists $i,j$ such that $1 > s_i \geq s_j > 0$, then set $s_i' = s_i + \eps$ and $s_j' = s_j - \eps$. The sum term remains the same and the product decreases:
        \begin{align*}
            (1-(s_i+\eps))(1-(s_j-\eps)) &= (1-s_i -\eps)(1-s_j+\eps) \\ 
            &= (1-s_i)(1-s_j) + \eps (s_j - s_i) - \eps^2 \\ 
            & < (1-s_i)(1-s_j),
        \end{align*}
        where the last line follows because $s_j-s_i \leq 0$ and $\eps > 0$. The other factors in the product are unchanged, and thus $L(\mbs)$ decreases.
        \item If there exists at most one coordinate $i$ with $s_i \in (0,1)$. Setting $s_i' = s_i + \eps$ decreases $L(\mbs)$:
        \begin{align*}
            s_i + \eps + C(1-(s_i+\eps)) &=s_i + \eps + C - C s_i - C\eps \\ 
            &= s_i + C(1-s_i) + \eps (1-C) \\ 
            &< s_i + C(1-s_i),
        \end{align*}
        where the last line follows because $C > 1$ and $\eps > 0.$
    \end{itemize}
    Thus, $\mbs^* = [1, 0,\ldots, 0]$.
    With $a \geq 1$, however, we have $R([1,0, \ldots, 0]) = C > 1 = L([1, 0, \ldots, 0])$, which is a contradiction for being in the left branch of the max.
    
    Combining these directions, we conclude that the minimizing $\mbs^*$ occurs at equality of $L(\mbs^*) = R(\mbs^*).$
\end{proof}

\subsection{Proof of \Cref{lemma:shape}}\label{app:shape-proof}
\begin{lemma*}[Characterizing the minimizer of $g$]
    Let $\mbs^*$ denote the minimizing vector $\mbs^* := \arg\min_{\mbs\in[0,1]}g(\mbs).$ Then $\mbs^*$ is either a designated uniform committee equilibrium or a symmetric committee equilibrium (\Cref{def:designated-uniform,def:symmetric}).
\end{lemma*}
\begin{proof}
    We perform this analysis using \eqref{prog:g-constrained}. First, we will show that if $s_i^* < 1, \forall i \in [n]$, then the shape of the optimizer is symmetric. Let $A = \{1, \ldots, a\}, H= \{a+1, \ldots, n\}$ denote the attacker and honest controlled indices and 
    \begin{align*}
        \pi_A &= \prod_{i\in A} (1-s_i), \qquad \pi_H = \prod_{i\in H}(1-s_i)
    \end{align*}
     denote the probability of the attacker and honest sets having no deliveries respectively. First we show that all honest players with non-zero mass have the same probability.
    
    Assume, towards a contradiction, that there exists $i, j \in H$ such that $s_{i} > s_j$. Then consider $\mbs'$ where all indices maintain their corresponding value in $\mbs$ \textit{except} $i, j$ and 
    \begin{align*}
        s_{i}' = s_{i} + \eps, \qquad s_j' = s_j - \eps,
    \end{align*} 
    for some small $\eps$.
    This transformation preserves the order and the sum of $\mbs$. For the product term $\pi_H$ each of the non $i,j$ indices are the same and 
    \begin{align*}
        (1-s_i')(1-s_j') &= (1-(s_i+\eps))(1-(s_j-\eps)) \\ 
        &= (1-s_i)(1-s_j) + \eps (s_j - s_i) - \eps^2 \\ 
        &< (1-s_i)(1-s_j) \tag{by $s_j < s_i, \eps>0$}.
    \end{align*}
    Thus, $C \pi_H$, the objective of \eqref{prog:g-constrained}, decreases. Examining the equality constraint of \eqref{prog:g-constrained}, we have
    \begin{align*}
        \sum_{i\in[n]} s_i + C \pi_A \pi_H = C \pi_H \implies \sum_{i\in[n]} s_i +C\pi_H(\pi_A - 1) = 0.
    \end{align*}
    Given $\pi_H$ decreased but the sum is preserved, we need to restore feasibility. Choose an attacker index $k \in A$ and fix all other values of $\mbs$. Then the equality constraint becomes
    \begin{align}\label{eq:affine-in-sk}
        s_k + \sum_{\substack{i \in [n]\\ i \neq k}} s_i + C \pi_H \bigg(\prod_{\substack{i \in [a]\\ i \neq k}} (1-s_i)\bigg) (1-s_k) = 0 
    \end{align}
    Notice that \eqref{eq:affine-in-sk} is affine in $s_k$. Thus, given a new value of $\pi_H$, we can solve this linear equation by increasing some corresponding value of $s_k$ to the unique value that restores the feasibility. With the reduced objective and the constraint satisfied, the transformation decreased the optimal value of \eqref{prog:g-constrained}, which is a contradiction. 
    
    We have that $s_{i} = s_{j},  \forall i,j \in H, s_i ,s_j > 0.$ Intuitively, this argument pushes probability mass up from the smallest values of $H$ to the higher probability members, concentrating the mass symmetrically among them. Thus we have $s_{a+1} = s_{a+2} = \ldots = s_{k},$ where $k$ is the committee size of the symmetric equilibrium. Next, we show that the attacker indices also share that same value.
    
    Assume, again towards a contradiction, that there exists $i \in A$ such that $s_i > s_{a+1}$ (i.e., there is an attacker index with a higher probability that the \textit{first} honest index). Recall that we are in the case where $s_i < 1, \forall i$, so there can be no deterministic deliverers. Then consider $\mbs'$ where we transform $s_{a+1}' = s_{a+1}' + \eps$. The objective $C \pi_H$ decreases. Again, we turn to feasibility. Since \eqref{eq:affine-in-sk} is affine in $s_i$, we adjust $s_i$ to restore equality with $0$. Again, we have reduced the value of the conjectured optimizer, and we have a contradiction. Intuitively, this step equalizes the attacker values with the honest ones because any probability mass that isn't evenly spread over all indices is being ``wasted'' on the attacker. So far, we know that if $s_i>1$ for all $i$, then the shape of the optimizer is a symmetric committee equilibrium $s_1 = s_2 = \ldots = s_{a+1} = \ldots s_k$  (\Cref{def:symmetric}).

    Now consider the case that there exist some deterministic deliverers in $\mbs$ (which, by the ordering constraint, implies $s_1=1$). First, we show that for all $i,j$ such that $s_i,s_j \in (0,1)$, $s_i=s_j$. In words, this is the fact that all the non-deterministic players have the same probability. Notice that with $s_1=1$, the objective remains $C \pi_H$, but the constraint changes to $\sum_{i\in [n]} s_i  = C \pi_H$. This case is much simpler. For any $i,j \in [n]$ if $s_i > s_j$ and $s_i,s_j \in (0,1)$, perform the transformation 
    \begin{align*}
        s_i' = s_i + \eps, \qquad s_j' = s_j - \eps.
    \end{align*}
    Thus sum is preserved and the objective $C \pi_H$ is weakly decreasing (it is strictly decreasing only if $j \in H$. Thus, the equality can be preserved and the objective is weakly lower, which is a contradiction. So all non-deterministic players have the same value.

    Lastly, we show that there is \textit{only one} deterministic deliverer, thus $s_1=1 \implies s_i < 1, \forall i \in \{2, \ldots, n\}.$ Assume towards a contradiction, that there are multiple players with $s_i = 1$. Let $k$ denote the largest index such that $s_k=1$ and consider the modified equilibrium $\mbs'$, where all indices are the same except $s_k' = 1-\eps$. This reduces the sum, but doesn't change the objective.  Using \eqref{eq:affine-in-sk}, we can solve for a value $\delta$ such that $s_{a+1}' = s_{a+1} + \delta$ restores the feasibility and strictly increases the value of $s_{a+1}$, which reduces the objective and leads to a contradiction. Intuitively, this step exploits the fact that multiple deterministic deliverers doesn't help because $\pi_A$ is already $0$ with just $s_1=1.$ That probability mass can be shifted down to reduce $C \pi_H$ instead. Thus, if there exists a deterministic deliverer, the shape of the minimizer is a designated uniform committee equilibrium (\Cref{def:designated-uniform}).
\end{proof}

\subsection{Proof of \Cref{lemma:designated-implementable}}\label{app:designated-implementable-proof}
\begin{lemma*}[All designated uniform committee equilibria are implementable]
    Given a designated uniform committee equilibrium parameterized by committee size $k\leq n-1$ and a mixing probability $s$, there exists a payment rule that implements the equilibria and achieves the minimal protocol loss of $1 + ks$.
\end{lemma*}
\begin{proof}
    We construct a payment rule that satisfies the following properties: (i) only pay if a proof is delivered, and (ii) always pay $1+ks$ no matter how many proofs are delivered.
    We achieve property (i) by conditioning all payments on player 1 delivering. Since player 1 is designated with $s_1=1$, this doesn't impact the equilibrium conditions. The committee members (indices $\{2, \ldots, k+1\}$) mix with uniform probability $s$ and have a prize of size $\frac{ks}{1-(1-s)^k}$ split evenly among the committee members who deliver. We formalize this below.
    Notice that the L1-norm $||\mb{d}||_1$ counts the number of deliveries. Then consider the following payment rule 
    \begin{align*}
        p_1(\mb{d}) &= 
            \begin{cases}
                0 & \text{if } d_1 = 0\\
                1 + ks & \text{if } ||\mb{d}||_1 = 1 \\
                1-\frac{(1-s)^k ks}{1-(1-s)^k} & \text{if } ||\mb{d}||_1 > 1 \\
            \end{cases} \\ 
        p_{i>1}(\mb{d}) &= 
            \begin{cases}
                0 & \text{if } d_1 = 0\\ 
                \frac{ks}{1-(1-s)^k} \cdot \frac{1}{||\mb{d}||_1-1} & \text{if } d_i = 1, i \in \{2, \ldots, k+1\}.
            \end{cases}
    \end{align*}
    We can quickly confirm that this is an equilibrium. Consider player 1's expected payoff from delivering,
    \begin{align*}
        \E[p_1(\mb{d})] &= (1+ks) \cdot (1-s)^k + \left(1-\frac{(1-s)^k ks}{1-(1-s)^k}\right) \cdot (1-(1-s)^k)\\
        &= 1.  
    \end{align*}
    The other players also have the same expected payoff from delivering
    \begin{align*}
        \E[p_{i>1}(\mb{d})] &= \sum_{j=0}^{k-1} \left(\frac{ks}{1-(1-s)^k} \cdot \frac{1}{j+1} \cdot \Pr[X = j]\right) = 1,
    \end{align*}
    where $X \sim \text{Binomial}(k-1,s)$. Thus $[1, s, \ldots , s, 0, \ldots,0]$ is an equilibrium. For any $||\mb{d}||_1 > 1$, the protocol pays $1+ks$:
    \begin{align*}
        1-\frac{(1-s)^k ks}{1-(1-s)^k} + \frac{ks}{1-(1-s)^k} = 1+ks.
    \end{align*}
    Similarly, if $||\mb{d}||_1 = 1$, the protocol pays $1+ks$.
\end{proof}

\subsection{Proof of \Cref{thm:payment-rule-reduction}}\label{app:payment-rule-reduction-proof}
\begin{theorem*}[Payment rule reduction]
    Given any payment rule, $p$, and corresponding equilibrium, $\mb{s}\in\mcs(p)$, there exists a modified payment rule, $p'$, and a modified equilibrium, $\mb{s}'\in\mcs(p')$, such that the modified equilibrium is designated and
    \begin{align*}
        \ell(p',\mb{s}') \leq \ell(p,\mb{s}) + 1.
    \end{align*}
\end{theorem*}
\begin{proof}
Without loss, re-index the entries of $\mb{s}$ to be decreasing,
\begin{align*}
    1 \geq s_1 \geq s_2 \geq \ldots \geq s_n \geq 0.
\end{align*}
If $s_1=1$, the equilibrium is already designated and the claim holds trivially. Thus we consider if $s_1 < 1$. We assign $s'_1 \rightarrow 1$ and keep all remaining entries of $\mb{s}$. That is, 
\begin{align*}
    \mb{s}' = [1, s_2, s_3, \ldots, s_n].
\end{align*}
To construct $p'$, we start by paying player 1 exactly $1$ if they deliver and $0$ otherwise, independently of the rest of the players' actions. For the remaining players, their payments under $p$ could depend on the realized action of player 1, $d_1$. Under $p'$, we pay them exactly the average of their payments weighted by $s_1$, but \textit{only if} player 1 does deliver. More formally, we set
\begin{align*}
    p_i'(\mb{d}) = \begin{cases} p_i(\mb{d} | d_1 = 0) \cdot (1-s_1)  + p_i(\mb{d} | d_1 = 1) \cdot s_1 & \text{if } d_1 = 1 \\ 
    0 & \text{otherwise}.
    \end{cases}
\end{align*}
Under $\mb{s}'$, the honest player 1 will \textit{always} deliver, so $s_{i>1}$ maintains the same expected payments and thus remains an honest equilibrium, i.e., $s'\in\mcs(p')$. Furthermore, $p'$ has the following properties: a) its total payment is always in $\big\{0,\sum_{i \in [n]} s'_i\big\}$ for all delivery outcomes $\mb{d}$, and b) it never pays when there is no delivery.

With $p',\mb{s}'$ defined, consider the loss of any attacker set and delivery decision, $A, \mb{d}_A$, under the original $p,\mb{s}$. Again, we bound this loss below by the maximum of the attacker being \textit{fully honest} and the attacker \textit{fully non-delivering} with the $a$ largest values of $s_i$. The loss under an honest attacker is \textit{at least} the total payment, which is at least $\sum_{i \in [n]} s_i$ since for $\mb{s}$ to be an equilibrium, expected payments must cover expected proving costs. Furthermore, the loss under the fully non-delivering attacker is \textit{at least} the penalty term $C \cdot \prod_{i=a+1}^n (1-s_i)$ even if there are no payments. Together, we have 
\begin{align}\label{eq:lem-1-lower-bound}
    \ell(p,\mb{s}) \geq \max \bigg\{\sum_{i \in [n]} s_i, C \cdot \prod_{i =a+1}^n (1-s_i)\bigg\}.
\end{align}
Now consider the loss of the new payment rule $p'$ under the modified equilibrium $\mb{s}'$. We break the possible attacker actions into two cases. If the attacker corrupts the designate and doesn't deliver (i.e., $1\in A$, $a_1=0$), then $p'$ pays nothing and the highest possible loss occurs when the attacker controls the largest values $s_1,\ldots,s_a$, resulting in a loss of $C\cdot \prod_{i={a+1}}^n (1-s'_i)$. Otherwise, the designate delivers and cost is always $\sum s'_i.$ Thus
\begin{align}\label{eq:lem-1-p-prime}
    \ell(p',\mb{s}') = \max\bigg\{\sum_{i \in [n]} s_i', C \cdot \prod_{i=a+1}^n (1-s'_i)\bigg\}.
\end{align}
The first term of \Cref{eq:lem-1-p-prime} is $1-s_1$ larger than the first term of \Cref{eq:lem-1-lower-bound}. The second terms of \Cref{eq:lem-1-lower-bound} and \Cref{eq:lem-1-p-prime} are equal if $a\geq 1$. It follows that 
\begin{align*}
    \ell(p', \mb{s}') \leq \ell(p,\mb{s}) + 1. 
\end{align*}
\end{proof}

\subsection{Proof of \Cref{lemma:symmetric-anonymous}}\label{app:symmetric-anonymous-proof}
\begin{lemma*}[Symmetric payment rule reduction]
    Any payment rule $p$ that implements a symmetric equilibrium $\mbs$ with a committee size of $k$ can be transformed into an anonymous payment rule $p'$ with the following form:
    \begin{align*}
        p'_i(\mb{d}) &= 
            \begin{cases}
                f_t/t & \text{if } d_i = 1, ||\mb{d}||_1 = t, i \leq k \\ 
                0 & \text{otherwise}.
            \end{cases} 
    \end{align*}
    where $f_t$ is a fixed total prize that the protocol pays under the event that there are exactly $t$ proofs delivered. The symmetric vector is still an equilibrium under the modified payment rule $\mbs \in \mcs(p')$ and $p'$ has weakly lower cost $\ell(p',\mbs) \leq \ell(p,\mbs)$.
\end{lemma*}
\begin{proof}
    Suppose $p$ implements $\mbs = (s,\ldots,s, 0, \ldots, 0)$ for $s\in(0,1)$ where $k \in [n]$ denotes the number of committee members. First, let $p'_i(\mbd) = 0, \forall i > k$. This weakly lowers the loss of the protocol because the players with index $i > k$ deliver with probability 0 and thus any positive payments to them only increase the cost to the protocol without reducing the probability of a penalty. Further, these payments don't interact with the equilibrium conditions of the $k$ committee members because the index $i > k$ players are playing the pure strategy of never delivering and they will continue to not deliver. 

    Now, we only consider the payments to the committee members with indices $i \leq k$, who are playing the symmetric strategy $s_i = s$. Partition the outcomes of $\mbd$ into sets based on the count of proofs delivered $t=||\mbd||_1.$ Notice that for each player $i$ in the committee, the probability of each outcome that results in exactly $t-1$ other deliveries is identical: $s^{t-1}(1-s)^{k-t}$. Further, the probability of each outcome that results in exactly $t$ other deliveries is: $s^t(1-s)^{k-t-1}$. Thus, consider an intermediate payment rule $\hat{p}$, where the payment received by player $i$ in the case that there are exactly $t$ deliveries is,
    \begin{align*}
        \hat{p}_i(\mbd) = \begin{cases}
            p_{i,t,1} & \text{if } i \leq k, d_i = 1, ||\mbd||_1 = t \\
            p_{i,t,0} & \text{if } i \leq k, d_i = 0, ||\mbd||_1 = t,
        \end{cases}
    \end{align*}
    where $p_{i,t,1}$ and $p_{i,t,0}$ are the average payment that player $i$ receives under the original payment rule $p$ if they deliver or don't deliver respectively and there are a total of $t$ deliveries:
    \begin{align*}
        p_{i,t,1} &= \frac{1}{\binom{k}{t-1}} \sum_{\substack{\mbd : ||\mbd||_1 = t\\d_i =1}} p_i(\mbd), \quad p_{i,t,0} = \frac{1}{\binom{k}{t}} \sum_{\substack{\mbd : ||\mbd||_1 = t\\d_i=0}} p_i(\mbd).
    \end{align*}

    Under the intermediate payment rule $\hat{p}$, player $i$ has exactly the same expected payment and thus $s$ continues to be an equilibrium strategy. Next, we average the payments \textit{over all committee members}. Define $\bar{p}_{t,1}$ and $\bar{p}_{t,0}$ to be the average payments for delivering and not delivering among the committee members conditioned on exactly $t$ deliveries:
    \begin{align*}
        \bar{p}_{t,1} = \frac{1}{k} \sum_{i=1}^k p_{i,t,1}, \quad \bar{p}_{t,0} = \frac{1}{k} \sum_{i=1}^k p_{i,t,0}, 
    \end{align*}
    If we pay all of the committee members according to $\bar{p}$ (e.g., if there are $t$ deliveries, pay each of the deliverers $\bar{p}_{t,1}$ and each of the non-deliverers $\bar{p}_{t,0}$), then we have exactly the same payments as $p$. Since the equilibrium has not changed, we have $\ell(\bar{p}, \mbs) = \ell(p,\mbs)$. 

    To conclude, we perform one more reduction where we remove the payments made to non-deliverers. Let $\bar{f}_t$ be the \textit{aggregate} payments made to all players in the event of exactly $t$ deliveries: $\bar{f}_t = t \cdot \bar{p}_{t,1} + (k-t)\cdot \bar{p}_{t,0}$. Consider the payment rule where only the $t$ deliverers are paid exactly $\bar{f}_t /t$ and the non-deliverers are paid 0. This preserves the total payment made by the protocol, but might make delivering too attractive. In particular, the expected reward from delivering is now
    \begin{align*}
        V := \sum_{t=1}^{k} \binom{k-1}{t-1} s^{t-1}(1-s)^{k-t} \frac{\bar{f}_t}{t}. \
    \end{align*}
    Our concern about overpaying for deliverers is encoded by $\bar{f}_t/t = \bar{p}_{t,1} + \frac{k-t}{t} \bar{p}_{t,0} \geq \bar{p}_{t,1}.$ That is, when we pay $\bar{f}_t/t$ only to the deliverers, they are weakly better off delivering under this payment rule than they were delivering under the $\bar{p}$ payment rule. In order to ensure $s$ is still an equilibrium, the expected value of payments conditioned on delivery must be \textit{exactly} 1, because we are not paying anything to the non-deliverers. Further, under the $\bar{p}$ rule, we had the following indifference condition:
    \begin{align*}
        \sum_{t=1}^{k} \binom{k-1}{t-1} s^{t-1}(1-s)^{k-t}(\bar{p}_{t,1} - \bar{p}_{t-1,0}) =1 ,
    \end{align*}
    where the expected value of delivering less the expected value of not delivering is exactly 1. Since $\bar{f}_t/t \geq \bar{p}_{t,1} \geq \bar{p}_{t,1} - \bar{p}_{t-1,0},$ we scale the final payment rule by dividing by $V$: 
    \begin{align*}
        p_i'(\mbd) &= \begin{cases}
            \frac{\bar{f}_t}{Vt} &\text{if } d_i = 1, ||\mb{d}||_1 = t, i \leq k \\ 
            0 & \text{otherwise}.
        \end{cases}
    \end{align*}
    Since the non-deliverers are not paid, we only have to check the indifference condition of delivery:
    \begin{align*}
        \sum_{t=1}^k \binom{k-1}{t-1} s^{t-1}(1-s)^{k-t} \frac{\bar{f}_t}{Vt} = V/V = 1.
    \end{align*}
    Since $V \geq 1$, $\frac{\bar{f}_t}{t}$ preserves the exact payments of $p$, and $\mbs$ is still an equilibrium, we have $\ell(p', \mbs) \leq \ell(p, \mbs)$. To match the final form of the lemma statement, simply let $f_i = \bar{f}_i / V$.
\end{proof}

\subsection{Expanded form of \eqref{prog:lp-symmetric}}\label{app:lp-symmetric-expanded}
\begin{align}
    &\min_{f_1,\ldots f_n \geq 0} t \tag{LP1}\\
    \text{s.t., \;\;} & \sum_{i=0}^h \binom{h}{i}s^{i}(1-s)^{h-i}f_{a+i} \leq t\tag{attack delivers $a$}\\ 
    & \sum_{i=0}^h \binom{h}{i}s^{i}(1-s)^{h-i}f_{a-1+i} \leq t\tag{attack delivers $a-1$}\\
    &\vdots \tag{attacker delivers with $\{a-2, a-3, \ldots, 1\}$} \\
    &\sum_{i=1}^h \binom{h}{i}s^{i}(1-s)^{h-i}f_{i} + C(1-s)^h \leq t\tag{attack delivers $0$}\label{eq:attack-delivers-zero}\\
    &\sum_{i=1}^{n} \binom{n-1}{i-1}s^{i-1}(1-s)^{n-i} \frac{f_i}{i} = 1 \tag{equilibrium condition}
\end{align}

\subsection{Proof of \Cref{lemma:symmetric-implementable}}\label{app:symmetric-implementable-proof}
\begin{lemma*}[Optimal symmetric implementability]
    For a given $(h,a,n)$, a symmetric minimizer of $g$ is implementable if 
    \begin{align*}
        C \geq \frac{hn(h+1)^{n-1}}{(h+1)^a - 1}.
    \end{align*}
\end{lemma*}
\begin{proof}
     \eqref{prog:linear-system} can be solved by setting $f_{h+1} = C\left(\frac{1-s}{s}\right)^h$ (from the first two equations in the system), and then recursively solving $f_{h+i}$ by substituting $f_{h+j}$ for $j <i$ into the equation. Algebraically, this yields the following relation:
    \begin{align*}
        f_{h+i} &= C \left(\frac{1-s}{s}\right)^h \cdot \sum_{j=0}^{i-1}(-1)^j \binom{h+j-1}{j} \cdot \left(\frac{1-s}{s}\right)^j.
    \end{align*}
    
    Since the values of $f_{h+1}, f_{h+2}, \ldots, f_{h+a}$ are fully determined, the feasibility check needs to confirm that all of their values are positive. Thus we need to evaluate the sign of the alternating sum. The magnitude of the $j^{th}$ term of the sum is $\binom{h+j-1}{j} (\frac{1-s}{s})^j$. We will consider the pair-wise ratio of the odd (negative) and even (positive) terms. This ratio is
    \begin{align*}
        \frac{\binom{h+(j+1)-1}{j+1} \left(\frac{1-s}{s}\right)^{j+1}}{\binom{h+j-1}{j} \left(\frac{1-s}{s}\right)^j} = \frac{h+j}{j+1} \cdot \frac{1-s}{s},
    \end{align*}
    which is weakly decreasing in $j$, so it is maximized at $j=0$, which gives $h \cdot \frac{1-s}{s}.$ For the sum to be positive, this ratio must be less than 1 (i.e., the positive term is larger than the negative terms) and the monotonicity in $j$ ensures that the remaining terms are also less than 1, which yields
    \begin{align*}
        h \cdot \frac{1-s}{s} \leq 1 \implies s \geq \frac{h}{h+1}.
    \end{align*}
    Further, consider any value of $h \cdot \frac{1-s}{s} > 1.$ Then evaluate $f_{h+2}$:
    \begin{align*}
        f_{h+2} = C \left(\frac{1-s}{s}\right)^h \left(1-h\cdot \frac{1-s}{s}\right). 
    \end{align*}
    Then $h \cdot \frac{1-s}{s} > 1 \implies f_{h+2} < 0$. This implies that $s \geq \frac{h}{h+1}$ is both a necessary and sufficient condition for of $f_{h+i} > 0$ and thus the symmetric equilibrium to be implementable. From the equality condition on the minimizers of $g$ (\Cref{lemma:minimizer-g-equality}), we have 
    \begin{align*}
        ns + C(1-s)^{n} = C(1-s)^h \implies C = \frac{ns}{(1-s)^h(1-(1-s)^a)}.
    \end{align*}
    Evaluating this at the boundary point $s= \frac{h}{h+1},$
    \begin{align*}
        C &= \frac{n(\frac{h}{h+1})}{(1-(\frac{h}{h+1}))^h(1-(1-(\frac{h}{h+1}))^a)} \\ 
        &= \frac{hn(h+1)^{h+a-1}}{(h+1)^a-1},
    \end{align*}
    which is the desired bound.
\end{proof}

\subsection{Lottery payment rules}\label{app:lottery-payment-rules}

\begin{definition}[Lottery payment rule]\label{def:lottery}
    A lottery payment rule sets a prize with magnitude $P$, which is delivered at random to a single deliverer if there are proofs delivered. Under no deliveries, no payment is made. More formally, let $D \subseteq [n]$ denote the set of deliverers, where $|D| = ||\mbd||_1$. Choose the ``winner'' $w \in D$ index uniform randomly from the deliverers. Then the lottery payment rule $p$ is defined as
    \begin{align*}
        p_i(\mbd) = \begin{cases}
            P & \text{if } ||\mbd||_1 > 0,  i =w \\
            0 & \text{otherwise}.
        \end{cases}
    \end{align*}
\end{definition}

The following lemma shows that this simple and intuitive payment rule is a (multiplicative) $2-$approximation of optimal.

\begin{lemma}[Lottery payment rules are a $2-$approximation of symmetric $\opt$]\label{lemma:lottery-2-approx}
    If the minimizer of $g$ is a symmetric committee equilibrium $\mbs^*$ (\Cref{def:symmetric}), then there exists a lottery payment rule $p$ (\Cref{def:lottery}) with a loss that is a two-approximation of $\opt$ where $\mbs^*$ is still an equilibrium: $2 \ell(p, \mbs^*) \leq g(\mbs^*) \leq \opt.$   
\end{lemma}
\begin{proof}
Let $\mbs$ denote the symmetric minimizer of $g$. From the equality condition on the minimizers of $g$ (\Cref{lemma:minimizer-g-equality}), we have 
\begin{align}
    C(1-s)^h = ns + C(1-s)^{h+a} &\implies ns = C(1-s)^h \left(1- (1-s)^a\right) \nonumber\\
    & \implies C(1-s)^h = \frac{ns}{1-(1-s)^a} \label{eq:c-ns-relation}.
\end{align}
This is the optimizing value of $g(\mbs)$, 
\begin{align*}
    g(\mbs) &= \frac{ns}{1-(1-s)^a}.
\end{align*}
Choose a lottery prize of 
\begin{align}\label{eq:lottery-prize-size}
    P = \frac{ns}{1-(1-s)^n},
\end{align}
which ensures that all $n$ players are indifferent between delivering and not (hence can play a mixed strategy) and $\mbs$ is an equilibrium: 
\begin{align*}
	\sum_{i=1}^n \binom{n-1}{i-1} s^{i-1}(1-s)^{n-i} \cdot \frac{P}{i} = 1.
\end{align*}
Notice that against a lottery, the best action the attacker can do is not deliver by setting $d_i = 0, \forall i \in A$. Thus the lottery achieves a loss of
\begin{align*}
    \ell(p,\mbs) = P \left(1-(1-s)^h\right) + C (1-s)^h.
\end{align*}
Evaluating the relative error of $\ell(p,\mbs)$ against $g(\mbs)$ and rewriting $C(1-s)^h$ from \eqref{eq:c-ns-relation} gives
\begin{align*}
    \frac{\ell(p,\mbs) - g(\mbs)}{g(\mbs)} &= \frac{\frac{ns}{1-(1-s)^n} \left(1-(1-s)^h\right) + \frac{ns}{1-(1-s)^a} - \frac{ns}{1-(1-s)^a}}{\frac{ns}{1-(1-s)^a}} \\ 
    &= \frac{\big(1-(1-s)^a\big) \big(1-(1-s)^h\big)}{1-(1-s)^n}  \leq 1,
\end{align*}
where the last step uses $(1-p)(1-q) \leq (1-pq), \forall p,q\in[0,1)$. A relative error of 1 corresponds to a $2-$approximation of the optimal. 
\end{proof}

\subsection{Proof of \Cref{lem:a1}}\label{app:a1-proof}
\begin{lemma*}[For $a=1$, designated is always optimal]
    For all $C>1, h \geq 1$, if $a=1$ then the optimal designated equilibrium is always better than the optimal symmetric: $D_j^* < S_k^*$.
\end{lemma*}
\begin{proof}
    First, we simplify the constraint of $S_k^*$ with $a=1$,
    \begin{alignat*}{2}
        &(h_k + 1)s + C(1-s)^{h_k+1} &&=  C(1-s)^{h_k} \\
        &\implies (h_k+1)s &&= C(1-s)^{h_k}(1-(1-s)) \\ 
        &\implies (h_k+1)s &&= C(1-s)^{h_k}s \\
    \end{alignat*}
    If $s=0$, then $S_k^*=C$ and if $s > 0$, then we have $S_k^* = h_k+1$, which is minimized at $h_k=1$. Thus,
    \begin{align*}
        S_1^* &= \min \{C, 2\}.
    \end{align*}
    Next, we turn to the designated constraint fixing $a=1$. This gives us 
    \begin{align*}
        1 + h_j s = C(1-s)^{h_j},
    \end{align*}
    which has a unique solution $s^* \in (0,1)$. The value of the constraint is minimized with $h_j =1$, so we get
    \begin{align*}
        1+s = C(1-s) \implies s = \frac{C-1}{C+1}, \quad D_1^* = \frac{2C}{C+1}.
    \end{align*}
    For $C>1$, it is true that 
   \begin{align*}
       D_1^* = \frac{2C}{C+1} < \min\{C, 2\} = S_1^*.
   \end{align*}
\end{proof}

\subsection{Proof of \Cref{lemma:desig-small}}\label{app:desig-small-proof}
\begin{lemma*}[Minimizer of $g$ is designated for $C < 1+1/a$]
    For very small penalty values $C < 1+1/a$, the minimizing vector $\mbs$ is a designated equilibrium.    
\end{lemma*}
\begin{proof}
    For a symmetric equilibrium $\mbs$ to be optimal, we know from \Cref{lemma:minimizer-g-equality} that the following equality holds:
    \begin{align*}
        C(1-s)^h = ns + C(1-s)^{n} &\implies n = C(1-s)^h \frac{\left(1- (1-s)^a\right)}{s}.
    \end{align*}
    For all $s \in(0,1), h \geq 1$, we have $(1-s)^h < 1$ and $\left(1- (1-s)^a\right) \leq sa$, thus we can apply the equality constraint:
    \begin{align*}
        n &= C(1-s)^h \frac{\left(1- (1-s)^a\right)}{s} \\ 
        &< Ca \implies C >n/a.
    \end{align*}
    Since $n = a+h$, we also have $C > 1 + h/a$ which implies for all $h \geq 1$, $C > 1 + 1/a$. Since the minimizer is only symmetric or designated (\Cref{lemma:shape}) and the symmetric equilibrium isn't feasible for $C < 1 + 1/a$, the minimizer $\mbs$ must be designated.
\end{proof}

\subsection{Proof of \Cref{lemma:sym-big}}\label{app:sym-big-proof}
\begin{lemma*}[Minimizer of $g$ is symmetric for $C \to \infty$]
    For large penalty values $C \to \infty$, the minimizing vector $\mbs$ is a symmetric equilibrium. 
\end{lemma*}
\begin{proof}
    At $C \to \infty$, the protocol will ensure that it never faces the catastrophic loss of receiving no proofs. Thus, both the optimal designated and symmetric equilibria will converge to the same binary vector $[1,\ldots, 1, 0 \ldots, 0]$, where there are exactly $a+1$ ones. This will guarantee the protocol cost stays bounded at exactly $a+1$. This is just the limit behavior, but for any finite $C$, the optimal symmetric and designated equilibria, denoted $\mbs_{sym}, \mbs_{des}$ respectively, will have the shapes
    \begin{align*}
        \mbs_{sym} &= [\underbrace{1-\eps, \ldots, 1-\eps}_{a+1 \text{ copies}}, 0, \ldots, 0] \\ 
        \mbs_{des} &= [1, \underbrace{1-\eps', \ldots, 1-\eps'}_{a \text{ copies}}, 0, \ldots, 0]
    \end{align*}
    where there are $a+1$ copies of $1-\eps$ in the symmetric case and $a$ copies of $1-\eps'$ in the designated case (because the initial coordinate is already fixed at $1$). This allows us to fix attention to $h_j=h_k=1$. Then, the designated equality constraint (\Cref{lemma:minimizer-g-equality}) reduces to 
    \begin{align*}
        1+a(1-\eps') = C \eps' \implies \eps' = \frac{1+a}{C+a},
    \end{align*}
    and the loss of this designated equilibrium is $D_1^* = C\eps'$.
    Similarly, the symmetric version of the equality constraint simplifies to
    \begin{align}
        (a+1)(1-\eps) +C\eps^{a+1} &= C\eps \nonumber\\
        (a+1)(1-\eps) &= C\eps (1-\eps^a) \nonumber\\ 
        a+1 &= C\eps \big(1+\eps + \ldots +\eps^{a-1}\big) \tag{dividing $(1-\eps)$} \\
        \implies C &= \frac{a+1}{\eps \big(1+\eps + \ldots +\eps^{a-1}\big)},\label{eq:implicit-eps-c}
    \end{align}
    and the loss of this symmetric equilibrium is $S_1^* = C \eps.$ Thus it suffices to show that there exists a $C$ such that $\eps < \eps'$ and correspondingly $S_1^* < D_1^*$. Plugging in the designated value $\eps' = \frac{1+a}{C+a}$ to \eqref{eq:implicit-eps-c}, we get 
    \begin{align*}
        C &= \frac{C+a}{\big(1+\eps' + \ldots +\eps'^{a-1}\big)} \\ 
        &\geq \frac{C+a}{1+\eps'} \\
        &= \frac{C+a}{1 + \frac{1+a}{C+a}} = \frac{(C+a)^2}{C + 2a+ 1}.
    \end{align*}
    The inequality is true for all $C > a^2$. Now observe that the RHS of \eqref{eq:implicit-eps-c} is strictly decreasing in $\eps$, and we confirmed that evaluating it at $\eps'$ resulted in a value that was less than $C$. The optimal symmetric equilibrium solves this function exactly for $C$, so we must have $\eps < \eps'$ and thus $S^*_1 < D_1^*.$
\end{proof}

\subsection{Proof of \Cref{lemma:limit-ct}}\label{app:limit-ct-proof}
\begin{lemma*}[Limit behavior of $C_t$]
As $n \to \infty$ the value of $C_t$ at which the optimal symmetric equilibrium achieves a lower loss than the optimal designated equilibrium as a function of the honest proportion $\tau \in (0,1)$ is:
\begin{align*}
    \lim_{n \to \infty}C_t(\tau) = e^x, \text{ where } 1+x = e^{(1-\tau)x}.
\end{align*}
\end{lemma*}
\begin{proof}
    We first show that as $n \to \infty$, the minimizing equilibria at the boundary $C_t$ set $h_j = h_k = h$. That is both the designated and symmetric minimizers use the full honest set in the limit. First, we start with the designated case. For a given, $h, n, C$, the optimal designated equilibrium mixing probability, $s$, solves the following implicit equation
    \begin{align*}
        1+ (n-1)s =C(1-s)^h.
    \end{align*} 
    Let $n \to \infty$ for a fixed honest fraction $\tau \in (0,1)$, and observe that the LHS of the equation is increasing to infinity as $s \to 1$ and the RHS of the equation is decreasing to zero as $s \to 1$. Thus, in order for the constraint to hold, we have $s = \Theta(1/n)$. Define this constant as $t_d$ (for designated), and let $s = t_d/n$. Then, consider the possible committee honest committee size of $\beta \in (0,\tau]$ (which is $h_j$ in the discrete case). Then we can write the continuous version of the designated constraint 
    \begin{align}
        1+ (1-\tau + \beta)t_d &= C(1-t_d/n)^{\beta n}  \label{eq:continuous-designated}\\
        &\approx Ce^{-\beta t_d} \tag{by exponential bound}
    \end{align}
    Using the continuous version of the designated loss, 
    \begin{align*}
        D_\beta := 1+ (1-\tau + \beta)t_d \implies t_d = \frac{D_\beta -1}{1-\tau+\beta}.
    \end{align*}
    Plugging this into the exponential bound gives,
    \begin{align*}
        D_\beta = Ce^{-\beta \cdot \left(\frac{D_\beta -1}{1-\tau+\beta}\right)} \implies C = D_\beta\cdot  e^{(D_\beta-1) \cdot \left(\frac{\beta}{1-\tau+\beta}\right)}.
    \end{align*}
    Hence for a fixed $C$, we have that $\frac{\beta}{1-\tau+\beta}$ is strictly increasing in $\beta \in (0,\tau]$, so $D_\beta$ must be strictly decreasing to hold the constraint tight. Thus the minimizing $D_\beta$ occurs at the endpoint $\beta = \tau$. Plugging this in to \eqref{eq:continuous-designated}, we get a much simplified
    \begin{align}
        D_\tau =1+ t_d = Ce^{-\tau t_d} \label{eq:reduced-continuous-designated}
    \end{align}
    Turning to the symmetric constraint, we choose $s = t_s / n$ (for symmetric)w as the ansantz for the same reason as the designated case. Again, let $\beta \in (0,\tau]$ denote the proportion of honest committee employed (the $h_k$ value in \Cref{def:minimizing-losses-g}). The continuous version of the symmetric constraint is
    \begin{align*}
      (1-\tau+\beta)t_s + (1-t_s/n)^{(1-\tau + \beta)n}  &= C (1-t_s/n)^{\beta n} \\
      (1-\tau+\beta)t_s + Ce^{-(1-\tau+\beta)t_s}  &= C e^{-\beta t_s} \tag{by exponential bound}
    \end{align*}
    Using the continuous version of the symmetric loss, 
    \begin{align*}
        S_\beta :=  C e^{-\beta t_s} \implies (1-\tau+\beta)t_s + S_\beta e^{-(1-\tau)t_s} = S_\beta,
    \end{align*}
    and rewriting gives
    \begin{align*}
        S_\beta = \frac{(1-\tau+\beta)t_s}{1-e^{-(1-\tau)t_s}}. 
    \end{align*}
    Now consider the fact that at $C_t$, we have that $S_\beta = D_\tau = 1+ t_d$. Solving for $\beta$ gives
    \begin{align}
        &\frac{(1-\tau+\beta)t_s}{1-e^{-(1-\tau)t_s}} = 1 + t_d \nonumber \\
        &\implies \beta = \frac{(1+t_d)\big(1-e^{-(1-\tau)t_s}\big)}{t_s} - (1-\tau) \label{eq:beta-td-ts}.
    \end{align}
    Plugging this into $Ce^{-\beta t_s} = 1+t_d$, we have
    \begin{align*}
        C = (1+t_d)e^{(1+t_d)(1-e^{-(1-\tau)t_s}) - (1-\tau)t_s)}.
    \end{align*}
    The minimizing symmetric mechanism will choose $t_s$ to minimizes this expression over $t_s$ for a fixed target $D_\tau = 1 + t_d$. Taking the derivative of the exponent with respect to $t_s,$
    \begin{align*}
        \frac{\partial}{\partial_s}&\left((1+t_d)(1-e^{-(1-\tau)t_s}) - (1-\tau)t_s\right) \\&= (1-\tau)((1+t_d)e^{-(1-\tau)t_s} - 1). 
    \end{align*}
    Setting this equal to zero, we get 
    \begin{align*}
        &(1-\tau)((1+t_d)e^{-(1-\tau)t_s} - 1) = 0 \\ 
        &\implies e^{-(1-\tau)t_s} = \frac{1}{1+t_d}.
    \end{align*}
    Plugging this expression into our function for $\beta$ \eqref{eq:beta-td-ts}, we get
    \begin{align*}
        \beta &= \frac{(1+t_d)\big(1-\frac{1}{1+t_d}\big)}{t_s} - (1-\tau)\\
        &= \frac{t_d}{t_s} - (1-\tau).
    \end{align*}
    From $D_\tau = S_\beta$, we have that $e^{-\tau t_d} = e^{-\beta t_s}$ and correspondingly $\tau t_d = \beta t_s$. Plugging in the value of $\beta$ gives the final result,
    \begin{align*}
        \tau t_d &= \left(\frac{t_d}{t_s} - (1-\tau)\right)t_s \\
        &= t_d -(1-\tau)t_s \\ 
        & \implies (1-\tau)t_d = (1-\tau)t_s \; \text{and } \; t_d = t_s \;\text{ and } \; \beta = \tau.
    \end{align*}
    This derivation allows us to confirm that at the transition point, we have $S_\tau$, which simplifies to
    \begin{align*}
        S_\tau = \frac{t_s}{1-e^{-(1-\tau)t_s}}. 
    \end{align*}
    Combining this with \eqref{eq:reduced-continuous-designated} and the fact that $t_s = t_d$ (which we denote simply as $t$ hereafter), we have that $t$ satisfies the following equation:
    \begin{align*}
       1+t = \frac{t}{1-e^{-(1-\tau)t}} \implies 1+t = e^{(1-\tau)t}.
    \end{align*}
    Applying this equality to \eqref{eq:reduced-continuous-designated}, we have that 
    \begin{align*}
        1+t = Ce^{-\tau t} &\implies e^{(1-\tau)t} = Ce^{-\tau t} \implies C = e^t,
    \end{align*}
    which, combined with $1+t = e^{(1-\tau)t}$ is the limit of $C_t$ given in the lemma statement.
\end{proof}

\subsection{Figure showing discrete values of $C_t$ for various $n$}\label{app:fig-ct-limit-discrete}
\begin{figure}[H]
    \centering
    \includegraphics[width=\linewidth]{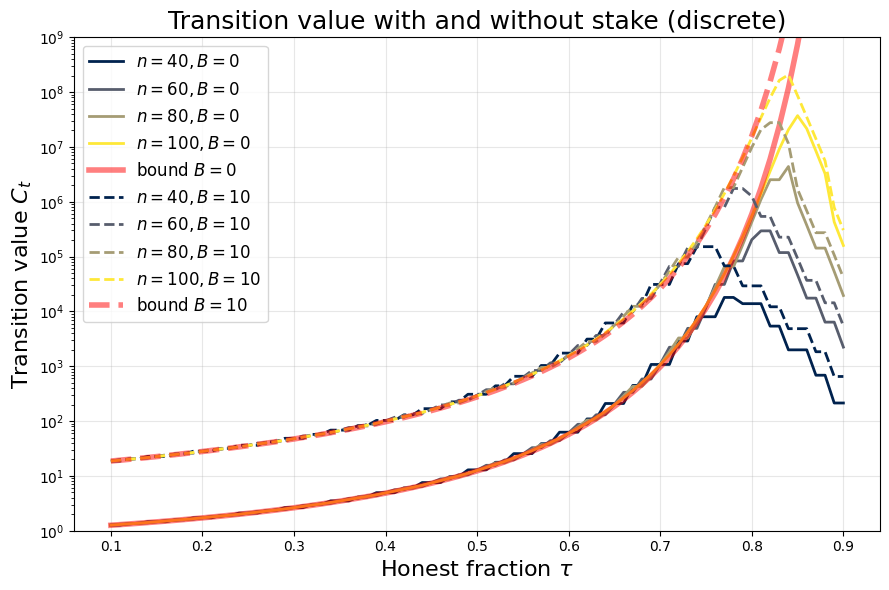}
    \caption{This plot shows $C_t$ (on log scale) as a function of $\tau$ for various values of $n$ and the limit behavior of $C_t$ as calculated in \Cref{lemma:limit-ct} (in red). In addition to the continuous limit (\Cref{fig:ct-limit}), this figure also includes the discrete values for  $n \in \{40,60,80,100\}.$ The dashed lines denote the corresponding values with stake $B=10$ (\Cref{subsec:negative-payments}) and the red dashed line is as calculated in \Cref{cor:limit-ct}.}
    \label{fig:ct-limit-discrete}
\end{figure}

\subsection{Proof of \Cref{lem:designated-loss-asymptotics}}\label{app:designated-loss-asymptotics-proof}
\begin{lemma*}[Asymptotic designated loss]
    The loss of the optimal designated equilibrium scales as $\mco\left(\frac{\log C}{\tau}\right)$. 
\end{lemma*}
\begin{proof}
    From \eqref{eq:reduced-continuous-designated}, we have
    \begin{align*}
        D_\tau =1+ t_d = Ce^{-\tau t_d}.
    \end{align*}
    Rewriting with $t_d = D_\tau -1$, we get
    \begin{align*}
        D_\tau = Ce^{-\tau (D_\tau-1)} \implies D_\tau e^{\tau D_\tau} = Ce^\tau.
    \end{align*}
    Multiplying through by $\tau$ gives an expression of the form $ye^y= x$, which is the definition of the Lambert W function:
    \begin{align*}
        \tau D_\tau e^{\tau D_\tau} = C\tau e^\tau \implies \tau D_\tau = W(C\tau e^\tau).
    \end{align*}
    Using the first two terms of asymptotic expansion of the Lambert W function, we have
    \begin{align*}
        D_\tau &= \frac{1}{\tau} \left(\ln (C \tau e^\tau)  + \mco(\ln\ln(C\tau e^\tau))\right) \\
        &= \frac{\ln C}{\tau} + \mco\left(\frac{\ln \ln C}{\tau}\right) \\
        &= \mco\left(\frac{\log C}{\tau}\right).
    \end{align*}
\end{proof}

\subsection{Proof of \Cref{corollary:lb-with-b}}\label{app:lb-with-b-proof}
\begin{corollary*}[Staked lower bound function $g_B(\mbs)$]
    Extending the lower bound of \Cref{lemma:lb}, consider the protocol parameterized with $(h,a,n,C)$ and with aggregate stake of $B$. Then define $g_B(\mbs)$ as,
    \begin{align*}
        g_B(\mbs) := \max\left\{\sum_{i\in [n]} s_i + C \prod_{i\in[n]} (1-s_i), C \prod_{i=a+1}^n (1-s_i) - B\right\}.
    \end{align*}
    Then for the solution to the follow program,
    \begin{align}
        \opt{4} := \min_{\mbs \in [0,1]^n} &g_B(\mbs) \tag{PROG4}\\
        \text{s.t., } &1\geq s_1 \geq s_2 \geq \ldots \geq s_n \ge 0 \nonumber
    \end{align}
    we have $\opt4 \leq \opt1$.
\end{corollary*}
\begin{proof}
    Recall that $\opt1$ (the solution to \eqref{prog:full} is the minimizing protocol loss over any payment rule that satisfies the equilibrium conditions. The left branch of the max is exactly the same as the proof of \Cref{lemma:lb}, where if the attacker plays honestly, the protocol incurs at least the cost of the honest equilibrium. The right branch of the max is similar to before, but now has a $-B$ linear term, meaning that in the case that the attacker is playing the non-delivery strategy, the protocol can detect it (as before, this detection is used to ensure that there is no positive payment in this case). But further, this attack detection can be used to slash the full bond posted by the provers to mitigate the expected cost of an attack. Thus this is the best that the protocol could hope to do, and is still a lower bound on the optimal loss. 
\end{proof}

\subsection{Lower bound $g$ analysis with stake $B$}\label{app:with-stake}
\begin{corollary}[Minimizer of $g_B$ occurs at equality]\label{cor:two-branches}
    As in \Cref{lemma:minimizer-g-equality}, the vector $\mbs^*$ that minimizes $g_B$ occurs where the two branches of the max are equal. More formally, we have
    \begin{align*}
        \sum_{i \in [n]} s_i^* + C\prod_{i \in [n]}(1-s_i^*) = C\prod_{i=a+1}^n(1-s_i^*) - B.
    \end{align*}
\end{corollary}
\begin{proof}
    This proof follows the exact procedure as that of \Cref{lemma:minimizer-g-equality}. We show that the optimizer can never be strictly in either branch. Recall that,
    \begin{align*}
        g_B(\mbs) := \max\left\{\sum_{i\in [n]} s_i + C \prod_{i\in[n]} (1-s_i), C \prod_{i=a+1}^n (1-s_i) - B\right\},
    \end{align*}
    and define $L,R$ as the left and right branches of that max:
    \begin{align*}
        L(\mbs) &:= \sum_{i\in [n]} s_i + C \prod_{i\in[n]} (1-s_i) \\
        R(\mbs) &:= C \prod_{i=a+1}^n (1-s_i) -B.
    \end{align*}
    Since only the right branch is a function of $B$, we only have to rule out the case that the optimizer is strictly in that branch. Assume, towards a contradiction, that $\mb{s}^*$ (the minimizing $\mb{s}$ of $g$) is strictly in the right branch of the max $L(\mbs^*) < R(\mbs^*)$. Notice that since $B$ is just a constant shift, the partial derivative of the right branch is independent of it:
    \begin{align*}
        \frac{\partial R}{\partial s_i^*} &= - C\cdot \prod_{\substack{ j=a+1 \\
         j\neq i}}^n (1-s_j^*) <0.
    \end{align*}
    Thus the same transformation used in the proof of \Cref{lemma:minimizer-g-equality} holds and the minimizer cannot be strictly in the right branch, regardless of the stake value $B$.
\end{proof}

\begin{corollary}[Characterizing the minimizer of $g_B$]\label{cor:shape}
    As in \Cref{lemma:shape}, let $\mbs^*$ denote the minimizing vector $\mbs^* := \arg\min_{\mbs\in[0,1]}g_B(\mbs).$ Then $\mbs^*$ is either a designated uniform committee equilibrium or a uniform committee equilibrium (\Cref{def:designated-uniform,def:symmetric}).
\end{corollary}
\begin{proof}
    The stake value $B$ only enters $g_B$ in the RHS of the equality constraint as a linear shift. Thus the feasible surface is distinct from $g$, but the transformations still hold. In particular, the equality constraint becomes
    \begin{align}
        s_k + \sum_{\substack{i \in [n]\\ i \neq k}} s_i + B+ C \pi_H \bigg(\prod_{\substack{i \in [a]\\ i \neq k}} (1-s_i)\bigg) (1-s_k) = 0,
    \end{align}
    which is still affine in $s_k$. Thus the transformations continue to reduce the objective and we can restore the feasibility by move back to the constraint surface. 
\end{proof}

\begin{corollary}[Limit behavior of $C_t$ with $B$]\label{cor:limit-ct}
As in \Cref{lemma:limit-ct}, with $n \to \infty$ the value of $C_t$ at which the optimal symmetric equilibrium achieves a lower loss than the optimal designated equilibrium as a function of the honest proportion $\tau \in (0,1)$ is:
\begin{align*}
    \lim_{n \to \infty}C_t(\tau) = e^x, \text{ where } 1+x +B= e^{(1-\tau)x}.
\end{align*}
\end{corollary}
\begin{proof}
    This proof mirrors the structure of that for \Cref{lemma:limit-ct}. On the designated side, the discrete implicit equation becomes
    \begin{align*}
        1+ (n-1)s =C(1-s)^h - B.
    \end{align*} 
    The corresponding continuous version is 
    \begin{align*}
        1+ (1-\tau + \beta)t_d &= Ce^{-\beta t_d} - B.
    \end{align*}
    The resulting optimal designated loss is also minimized at $\beta= \tau$, yielding the simplified
    \begin{align}\label{eq:reduced-ctous-designated-stake}
        D_\tau =1+ t_d = Ce^{-\tau t_d} -B
    \end{align}
    Taking the continuous version of the symmetric constraint also with $\beta= \tau$, we have
    \begin{align*}
        S_\tau = \frac{t_s+Be^{-(1-\tau)t_s}}{1-e^{-(1-\tau)t_s}}. 
    \end{align*}
    Because $S_\tau = D_\tau$ at the transition point $C_t$, we have  
    \begin{align*}
        1+ t= \frac{t+Be^{-(1-\tau)t}}{1-e^{-(1-\tau)t}} \implies 1 +t +B = e^{(1-\tau)t}.
    \end{align*}
    Combining with \eqref{eq:reduced-ctous-designated-stake}, we get
    \begin{align*}
        Ce^{-\tau t}= e^{(1-\tau)t} \implies C = e^t.
    \end{align*}
    This combined with $1+t+B = e^{(1-\tau)t}$ is the limit of $C_t$ given in the lemma statement.
\end{proof}

\begin{figure}[H]
    \centering
    \includegraphics[width=\linewidth]{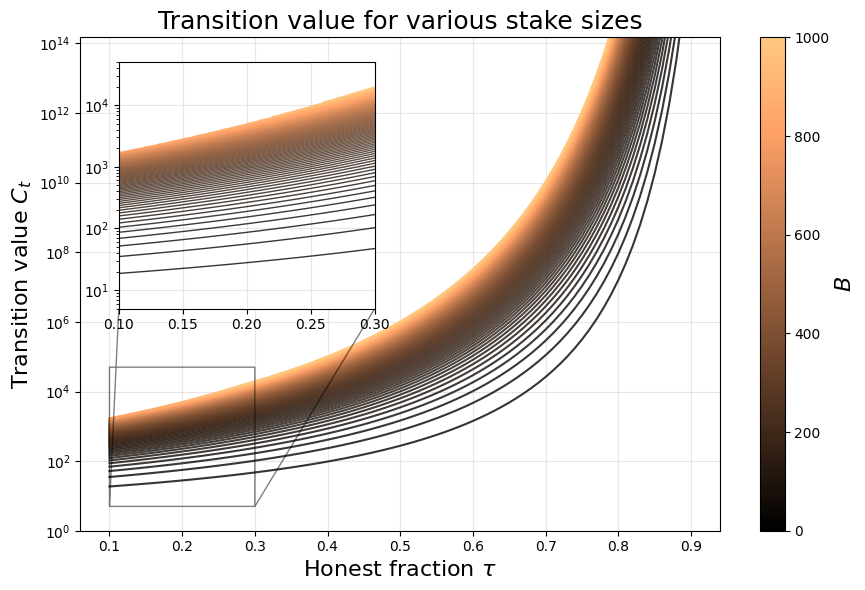}
    \caption{The value of $\lim_{n \to\infty} C_t(\tau)$ (\Cref{cor:limit-ct}) with various values of $B$ (indicated by the colorbar) as a function of $\tau$. Each point can be interpreted as, ``Given $\tau, B$, we plot the first value of $C$ such that designated has a lower cost that symmetric and thus the optimal is implementable.}
    \label{fig:c-over-b}
\end{figure}

\Cref{fig:c-over-b} shows the limit behavior of $C_t$ with stake (\Cref{cor:limit-ct}) for various values of $\tau,B$. As with \Cref{fig:ct-limit}, the area under the curve represents the values of $C$ such that the designated equilibrium is better than the symmetric on $g_B$, and thus the lower bound on the optimal cost is tight and implementable with a designated payment rule.

\subsection{Table for \Cref{ex:high-tau-c,ex:moderate-tau-c}}\label{app:stake-sensitivity-table}
\begin{table}[H]
\centering
\scriptsize
\begin{tabular}{rr rr r | rr rr r}
\toprule
\multicolumn{5}{c}{$n=100,\; a=33,\; C=10{,}000$} 
& \multicolumn{5}{c}{$n=200,\; a=100,\; C=20$} \\
\cmidrule(r){1-5} \cmidrule(l){6-10}
$B$ & $B/C$ & Desig. & Symm. & \% red. 
& $B$ & $B/C$ & Desig. & Symm. & \% red. \\
\midrule
0   & 0\%   & 10.62 & 10.12 & --     & 0    & 0\%   & 4.12 & 3.99 & --    \\
10  & 0.1\% & 9.78  & 9.79  & 3.3\%  & 0.02 & 0.1\% & 4.11 & 3.99 & --    \\
50  & 0.5\% & 8.32  & 9.64  & 17.8\% & 0.1  & 0.5\% & 4.09 & 3.99 & --    \\
100 & 1\%   & 7.48  & 9.56  & 26.1\% & 0.2  & 1\%   & 4.06 & 3.98 & --    \\
500 & 5\%   & 5.31  & 9.37  & 47.5\% & 1    & 5\%   & 3.81 & 3.97 & 4.5\% \\
\bottomrule
\end{tabular}
\caption{Optimal designated and symmetric costs under varying stake levels. \textbf{Left:} high $\tau$ and $C$ ($n=100$, $a=33$, $C=10{,}000$) leading to a large 47.5\% cost reduction from 5\% stake to penalty ratio (\Cref{ex:high-tau-c}). \textbf{Right:} lower $\tau$ and $C$ ($n=200$, $a=100$, $C=20$); even at $B/C=5\%$, the designated mechanism reduces loss by only 4.5\% (\Cref{ex:moderate-tau-c}). All \% reductions are relative to the symmetric cost at $B=0$.}
\label{tab:stake-comparison}
\end{table}

\subsection{Proof of \Cref{lem:staking-asymptotics}}\label{app:staking-asymptotics-proof}
\begin{lemma*}[Asymptotic loss with staking]
    The loss of the optimal designated equilibrium with staking amount $B$, scales as $$\mco\left(\frac{\log C - \log(\log C + B)}{\tau}\right).$$ Observe that this implies the following regimes
    \begin{alignat*}{2}
        B &= O(1) &&\implies loss = \mco\left(\log C / \tau \right) \\ 
        B &= \Theta(C/\textnormal{polylog } C) &&\implies loss = \mco\left(\log \log C/\tau\right)\\ 
        B &= \Theta(C) &&\implies loss = \mco\left(1/\tau\right).
    \end{alignat*}
\end{lemma*}
\begin{proof}
    As in the proof of \Cref{lem:designated-loss-asymptotics}, we start with the loss of the optimal designated equilibrium in the continuous setting but now with stake added:    
    \begin{align*}
        D_\tau =1+ t_d = Ce^{-\tau t_d} - B
    \end{align*}
    Rewriting with $t_d = D_\tau -1$, we get
    \begin{align*}
        D_\tau = Ce^{-\tau (D_\tau-1)} -B \implies (D_\tau+B) e^{\tau D_\tau} = Ce^\tau.
    \end{align*}
    Multiplying through by $\tau e^{\tau B}$ gives an expression of the form $ye^y= x$, which is the definition of the Lambert W function:
    \begin{align*}
        \tau (D_\tau+B) e^{\tau (D_\tau + B)} = C\tau e^{\tau(B+1)} &\implies \tau (D_\tau+B) = W(C\tau e^{\tau(B+1)}) \\
        &\implies D_\tau = \frac{1}{\tau}W(C\tau e^{\tau(B+1)}) - B.
    \end{align*}
    Using the first three terms of asymptotic expansion of the Lambert W function as $W(x) = \ln x - \ln \ln x + o(\ln \ln x)$, we have
    \begin{align*}
        D_\tau &= \frac{\ln C}{\tau} - \frac{\ln (\ln C +  B)}{\tau}+ o(1)\\
        &= \mco\left(\frac{\log C - \log(\log C +  B)}{\tau}\right),
    \end{align*}
    which is the expression in the lemma statement.
    Now let's consider each of the regimes. 
    \begin{itemize}
        \item $B=O(1)$: With a constant $B$, $\frac{\log(\log C + B)}{\tau}$ scales as $\log \log C / \tau$, thus the full scaling is $D_\tau = \mco(\log C / \tau).$ 
        \item  $B = \Theta(C/\text{polylog } C)$: Evaluating the second expression with $B = C / \log^p C$ for some $ p\geq 0$,
        \begin{align*}
             \log(\log C + B) &= \log(\log C + C / \log^p C) \\ 
             &= \log(C / \log^p C)+ O(1) \\
             &= \log C - p \log \log C + o(1).
        \end{align*}
        Thus the full scaling is $D_\tau = \mco(\log \log C / \tau)$.
        \item $B=\Theta(C)$: Evaluating the second expression with $B = \gamma C$ for some $\gamma \in (0,1)$, we have 
    \begin{align*}
        \log(\log C + B) &= \log(\log C + \gamma C) \\ 
        &= \log C + O(1).
    \end{align*}
    Thus the full scaling is $D_\tau = \mco(1/\tau).$
    \end{itemize}
\end{proof}

\section{Section 4 appendices}

\subsection{Discussion of tightness regimes}\label{app:regimes-discussion}
\begin{definition}[Regimes of conjectured optimal]\label{def:regimes}
    Given the conjectured minimizers $\hat{S}_k,\hat{D}_j$ (\Cref{def:minimizing-losses-implementable}) and the $g$ minimizers $S^*_k,D^*_j$  (\Cref{def:minimizing-losses-g}), four regimes are possible based on the shape of the equilibria and the tightness with the lower bound $g$. These regimes arise from the minimizer of $g$ being designated as $C \to 1^+$ (\Cref{lemma:desig-small}) and symmetric as $C \to \infty$ (\Cref{lemma:sym-big}).
    \begin{enumerate}
        \item \textbf{$D^*_j = \hat{D}_j < S^*_k$: designated and tight.} As long as the minimizers of $g$ are designated, we know that the lower bound is implementable (\Cref{lemma:designated-implementable}).
        \item \textbf{$S^*_k < \hat{D}_j < \hat{S}_k$: designated and not tight.} At $C > C_t$, the minimizer of $g$ becomes symmetric. However, it is possible that the symmetric minimizer is not implementable (c.f., \Cref{example:non-implementable}). In this regime, it is possible that $\hat{D}_j < \hat{S}_k$, so the best implementable designated equilibrium is better than the best implementable symmetric equilibrium.  
        \item \textbf{$S^*_k < \hat{S}_k < \hat{D}_j$: symmetric and not tight.} Similarly, it is possible that $\hat{S}_k < \hat{D}_j$ and thus the best implementable symmetric equilibrium is better than the best designated. However, there is still no solution to \eqref{prog:linear-system} with all $f_{i>h} \geq 0$, thus the bound is not tight.
        \item \textbf{$S^*_k=\hat{S}_k$: symmetric and tight.} Lastly, once $C$ is sufficiently large (lower bounded by \Cref{lemma:symmetric-implementable}), \eqref{prog:linear-system} will be implementable with non-negative payments, and the lower bound will again be tight.
    \end{enumerate}
     The transition points between the different regimes do not admit a general closed form (\Cref{remark:no-closed-form}) as they are the solutions to corresponding high-dimensional polynomial root finding problems.
\end{definition}

\begin{figure}[H]
    \centering
    \includegraphics[width=0.9\linewidth]{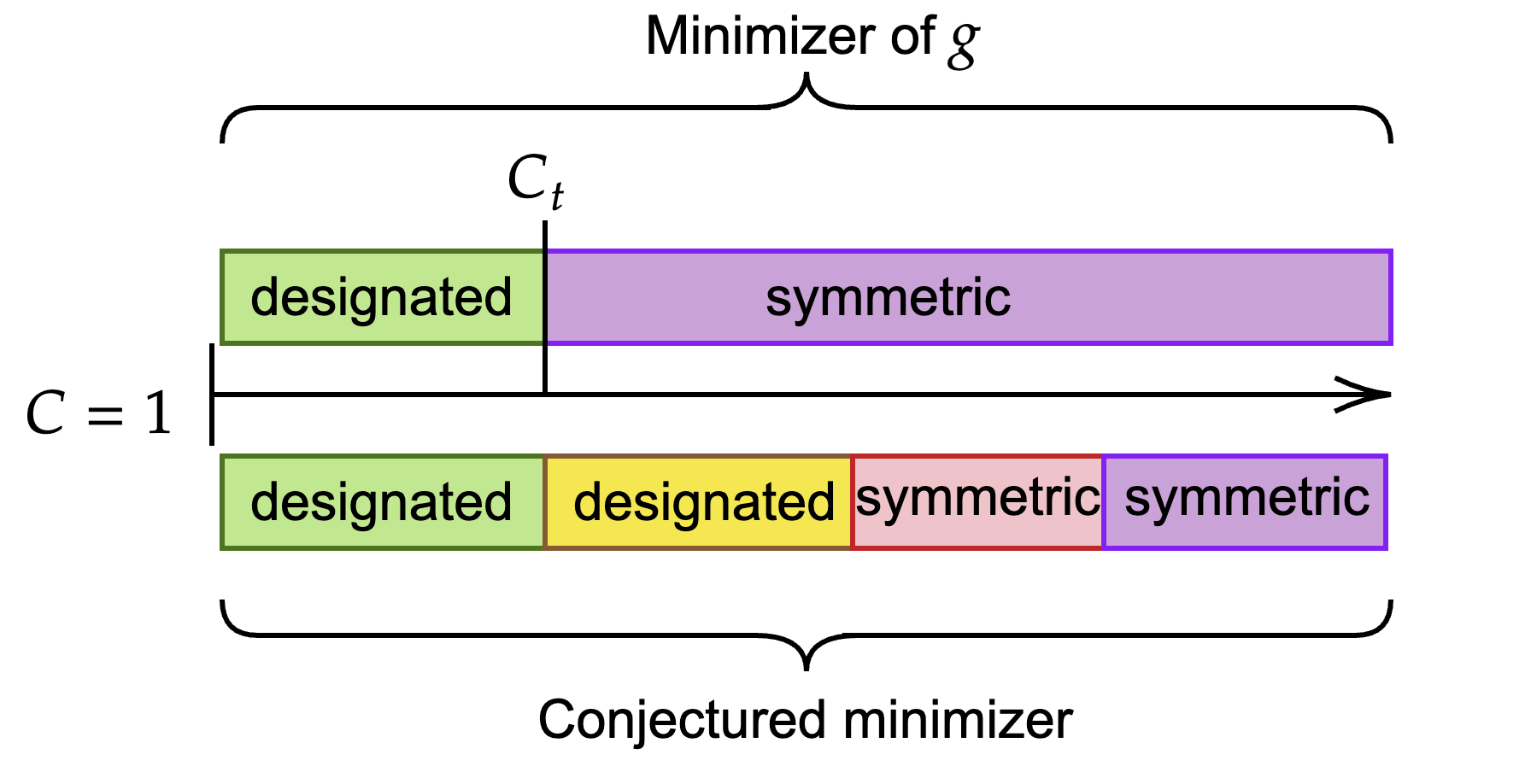}
    \caption{The four regimes (\Cref{def:regimes}) of the conjectured optimum shown pictorially. The designated and symmetric endpoints are colored the same as the lower bound to show that they are optimal.}
    \label{fig:regimes}
\end{figure}
\begin{figure}
    \centering
    \includegraphics[width=\linewidth]{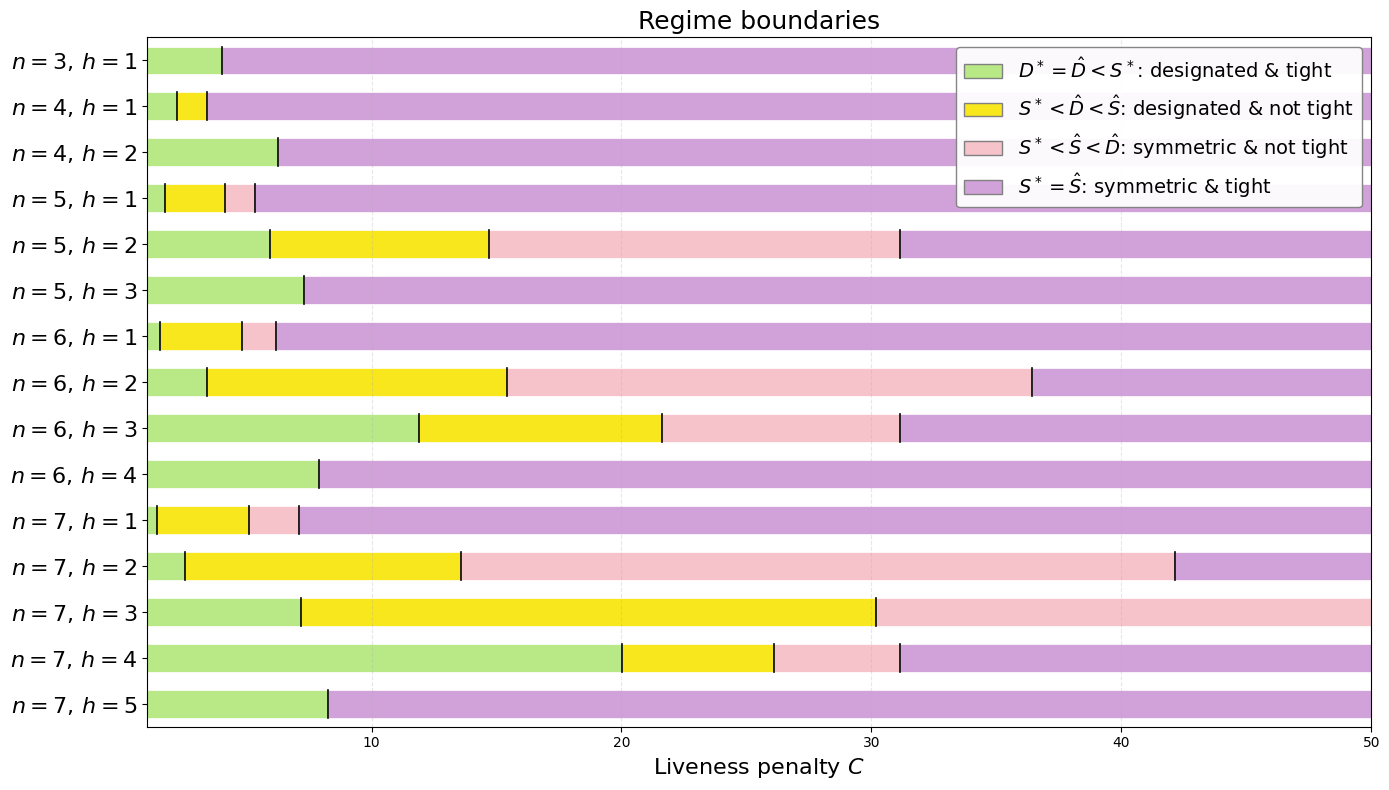}
    \caption{The numerical values of the four regimes (\Cref{def:regimes}) for various values of $\{(n,h) : n \in \{3, 4, \ldots, 7\}, h \in \{1, 2, \ldots, n-2\}\}$ as a function of $C$. We use the same color palette as \Cref{fig:regimes}.}
    \label{fig:regimes-actual}
\end{figure}

\Cref{fig:regimes} shows these pictorially as a function of $C$ increasing. 
\Cref{fig:regimes-actual} shows the numerical regimes in the discrete case for all values of $h < n-1$ for $n \in \{3, 4, \ldots, 7\}$. To further illustrate this point, consider the following example.

\begin{example}[$h=3,n=6$ regimes]\label{ex:h3n6}
    Let $h=3, n=6$ and consider the following regimes as a function of $C$:
    \begin{enumerate}
        \item $C\in[1,11.88]$ \textbf{Phase 1: designated and tight}.
        \item $C\in[11.88,21.60]$ \textbf{Phase 2: designated and not tight}.
        \item $C\in[21.60, 31.14]$ \textbf{Phase 3: symmetric and not tight}.
        \item $C\in[31.13, \infty)$ \textbf{Phase 4: symmetric and tight}.
    \end{enumerate}
\end{example}

\subsection{Counter examples for potential structure}\label{subsec:counter-examples}
This section examines many potential conjectures that would make the analysis of the lower bound function $g$ (\Cref{lemma:lb}) and of the conjectured optimal mechanism (\Cref{conj:minima}) more tractable. We show counter-examples to each conjecture, demonstrating the complexity of problem.

\begin{conjecture}[false]\label{conj:g-transition-at-same-committee}
    The transition point $C_t$ of $g$ (\Cref{def:transition}) always occurs with $h_j=h_k$.
\end{conjecture}
This conjecture would simplify the analysis of $D_j^*, S^*_k$ because we would set the values of the respective committees as the same value.
\begin{example}[Counter-example of \Cref{conj:g-transition-at-same-committee}]
    Let $h=14,n=19$. Then $C_t = 470.2382$ and we transition from a designated equilibrium with $k=19, h_k=14$ to a symmetric equilibrium with $k=12,h_k=7$.
\end{example}
Intuitively, what this counter-example shows is that we must consider all possible committee sizes for the candidate for the smallest value at which \textit{any} symmetric equilibria achieves a lower cost than \textit{any} designated equilibrium, because that is the point at which the lower bound ceases to be tight and achievable.

\begin{conjecture}[false]\label{conj:monotone-in-h}
    The shape of the conjectured optimal mechanism (\Cref{conj:minima}) has a single transition point as we increase $h\nearrow n-1$. 
\end{conjecture}

\begin{example}[Counter-example to \Cref{conj:monotone-in-h}]
    Consider $n=7, C=15$. Then the following are the solution to the conjectured optimal (\Cref{conj:minima}) as a function of $h$. 
    \begin{itemize}
        \item $h=1 \implies \mbs^* = (0.682, 0.682, 0.682, 0.682, 0.682, 0.682, 0.682)$ \\(symmetric full). 
        \item $h=3 \implies \mbs^* = (1, 0.393, 0.393, 0.393, 0.393, 0.393, 0.393)$ \\(designated full). 
        \item $h=5 \implies \mbs^* = (0.829, 0.829, 0.829, 0,0,0,0)$ \\(symmetric committee $k=3$). 
        \item $h=6 \implies \mbs^* = (1, 0.875, 0,0,0,0,0)$ \\(designated committee $k=1$). 
    \end{itemize}
\end{example}

Intuitively, this counter-example shows that the phase transitions of the shape are not monotone in $h$. In the geometric interpretation (c.f., \Cref{fig:surface}), this means that the optimizer is bouncing between multiple edges of the polytope as we increase the value of $h$.

\begin{conjecture}[false]\label{conj:monotone-in-n}
    The shape of the conjectured optimal mechanism (\Cref{conj:minima}) has a single transition point if we hold $h$ and increase $n$. 
\end{conjecture}

\begin{example}[Counter-example to \Cref{conj:monotone-in-n}]
    Consider $h=3, C=15$. Then the following are the solution to the conjectured optimal (\Cref{conj:minima}) as a function of $n$. 
    \begin{itemize}
        \item $n=4 \implies \mbs^* = (1, 0.875, 0, 0)$ \\(designated committee $k=1$). 
        \item $n=5 \implies \mbs^* = (0.829, 0.829, 0.829, 0,0)$ \\(symmetric committee $k=3$). 
        \item $n=6 \implies \mbs^* = (1, 0.411, 0.411,0.411,0.411,0.411,0.411)$ \\(designated committee $k=5$). 
        \item $h=15 \implies \mbs^* = (0.311, \ldots, 0.311)$ \\(symmetric committee $k=15$). 
    \end{itemize}
\end{example}

This counter example extends \Cref{conj:monotone-in-h} to consider holding $h$ fixed and considering the geometry of the conjectured optimal as we increase $n$. Just as in the previous counter-example, we see that the minimizer transitions from designated to symmetric edges of the polytope in a non-monotone way.

\begin{conjecture}[false]\label{conj:lp-structure-on-tight-constraints}
    The solution to the optimal implementable symmetric equilibrium \eqref{prog:lp-symmetric} results in a predictable set of constraints being tight and $f_i$ values being set to 0 (e.g., their is a threshold $i$ past which all of the \eqref{const:sym-lp-attacker} constraints are tight). 
\end{conjecture}

\begin{example}[Counter-example to \Cref{conj:lp-structure-on-tight-constraints}]
    Consider $n=16, h=8$. With $C = 3000$, we have,
    \begin{align*}
        f_1,f_2,f_3,f_4,f_5, f_8,f_9,f_{12},f_{15},f_{16} &= 0\\
        f_6,f_7, f_{10}, f_{11}, f_{13},f_{14} &>0.
    \end{align*}
    Further, the tight constraints \eqref{const:sym-lp-attacker} are $i=0,3,4,5,6,8$ and the slack constraints are $i=1,2,7$.

    Reducing to $h=7$, we have a fully different shape for the zero values of $f_i$ (differences are highlighted in red):
    \begin{align*}
        f_1,f_2,f_3,f_4,f_5, f_8, \textcolor{red}{f_{11}}, \textcolor{red}{f_{14}}, f_{15} &= 0\\
        f_6,f_7,\textcolor{red}{f_9}, f_{10},\textcolor{red}{f_{12}},f_{13},\textcolor{red}{f_{16}} &>0,
    \end{align*}
    and different slack constraints: $i=1,2,5$.
\end{example}

Intuitively, this example shows that there isn't a simple story about the set of tight constraints or the values of $f_i$ that the optimal symmetric equilibrium will set in the solution to the LP. 

\begin{conjecture}[false]\label{conj:conj-transition-at-same-committee}
    The transition from symmetric to designated of the conjectured optimal \Cref{conj:minima} always occurs with $h_j = h_k$ (similar to \Cref{conj:g-transition-at-same-committee}).
\end{conjecture}

\begin{example}[Counter-example to \Cref{conj:conj-transition-at-same-committee}]
    Consider $n=7, h=4$. The transition value $C$ for which $\hat{D}_j > \hat{S}_k$ occurs at $C\approx 26.093$, and we transition from a designated equilibrium with $h_j = 4, s=0.399$ to a symmetric equilibrium with $h_k=2, s=0.653$.
\end{example}

Just as in \Cref{conj:g-transition-at-same-committee}, we see that the size of the subcommittee is a key determinant in the optimal conjectured transition point. These two claims together imply that there is no simple way to determine any of the transition points of the conjectured minimizer as described in \Cref{fig:regimes}.
\balance

\end{document}